%% file: main.tex
\pdfoutput=1
\newif\ifarxiv
\arxivtrue

\ifarxiv
    \documentclass[acmsmall,screen,nonacm]{acmart}
\else
    \documentclass[acmsmall,screen]{acmart}
\fi

\usepackage{stmaryrd}
    \makeatletter
    \newcommand{\llangle}{\mathopen{\langle\!\langle}}
    \newcommand{\rrangle}{\mathclose{\rangle\!\rangle}}
    \makeatother
\usepackage{mathpartir}
\usepackage{microtype}
\usepackage{mathtools}
\usepackage{multirow}
\usepackage{siunitx}
\usepackage[table]{xcolor}
\usepackage{csquotes}
\usepackage{booktabs}
\usepackage{subcaption}
\usepackage{wrapfig}
\usepackage{hyperref}
\usepackage[inline]{enumitem}
\usepackage[capitalize,nameinlink]{cleveref}
\usepackage{amsthm}
    \theoremstyle{plain}
        \newtheorem{theorem}{Theorem}
        \newtheorem{lemma}{Lemma}
    \theoremstyle{definition}
        
        \newtheorem{property}{Property}
        \crefname{property}{Property}{Properties}
        \Crefname{property}{Property}{Properties}
    \theoremstyle{remark}
        \newtheorem*{remark}{Remark}

\usepackage{xspace}
\usepackage{xparse}
\usepackage{acro}
    \acsetup{patch/maketitle=false}
    \input{acro}
\usepackage{tikz}
    \usetikzlibrary{fit, positioning, arrows.meta}
    \tikzstyle{n}=[fill=none, font=\tiny,node distance=7pt, inner sep=1pt,minimum width=3mm, minimum height=3mm]
    \tikzstyle{m}=[fill=none, font=\tiny,node distance=5pt and 4pt, inner sep=1pt,minimum width=3mm, minimum height=3mm]
    \tikzstyle{o}=[fill=none, font=\tiny,node distance=2pt and 8pt, inner sep=1pt,minimum width=3mm, minimum height=3mm]
    \tikzstyle{a}=[-{Latex[length=3pt]}, semithick, densely dotted]
\usepackage{pgfplots}
    \usepgfplotslibrary{groupplots,fillbetween}
    \pgfplotsset{compat=1.18}

\input{defs}
\input{lst}

\AtBeginDocument{
    \setlength{\textfloatsep}{8pt plus 2pt minus 2pt} 
    \setlength{\intextsep}{8pt plus 2pt minus 2pt}   
    \setlength{\floatsep}{6pt plus 2pt minus 2pt}    
    \setlength{\abovedisplayskip}{3pt}
    \setlength{\belowdisplayskip}{3pt}
    \setlength{\abovedisplayshortskip}{0pt}
    \setlength{\belowdisplayshortskip}{3pt}
}

\setcopyright{cc}
\setcctype{by}
\ifarxiv
\else
    \acmJournal{PACMPL}
    \acmYear{2026} \acmVolume{10} \acmNumber{PLDI} \acmArticle{208}
    \acmMonth{6} \acmDOI{10.1145/3808286}
    \acmSubmissionID{pldi26main-p239-p}
    \received{2025-11-14}
    \received[accepted]{2026-04-03}
\fi



\begin{document}

\title{SSA without Dominance for Higher-Order Programs}

\author{Roland Leißa}
\orcid{0000-0002-2444-6782}
\affiliation{%
	\institution{University of Göttingen}
	\city{Göttingen}
	\country{Germany}
}
\email{roland.leissa@cs.uni-goettingen.de}

\author{Johannes Griebler}
\orcid{0009-0002-2147-9068}
\affiliation{%
	\institution{University of Göttingen}
	\city{Göttingen}
	\country{Germany}
}
\email{j.griebler@stud.uni-goettingen.de}

\input{abstract}

\maketitle
\acresetall
\input{intro}
\input{overview}
\input{sema}
\input{nest}
\input{trie}
\input{eval}
\input{relwork}
\input{concl}

\begin{acks}
    The authors thank the anonymous reviewers for their helpful feedback,
    Joachim Meyer for implementing MimIR's \texttt{regex} plugin,
    Guido Moerkotte for insights on link-cut trees,
    and Sebastian Hack and Michel Steuwer for their support of the MimIR project.
\end{acks}

\section*{Data-Availability Statement}

The artifact accompanying this paper---including all benchmarks, measurement data, the \lean mechanization of all lemmas and theorems, and the \mimir implementation used for the experiments---is publicly available on Zenodo~\cite{leissa_2026_19069679}.
The \mimir framework is developed at \url{https://mimir.github.io}, where all contributions from this paper have been integrated.

\bibliographystyle{ACM-Reference-Format}
\bibliography{dblp,other}

\end{document}

%% file: acro.tex
\newcommand{\newacro}[2]{\DeclareAcronym{#1}{short=#1,long=#2}}

\newacro{AD}{automatic differentiation}
\newacro{AST}{abstract syntax tree}
\newacro{BST}{binary search tree}
\newacro{CFG}{control-flow graph}
\newacro{CFA}{control-flow analysis}
\newacro{CSE}{common subexpression elimination}
\newacro{DAG}{directed acyclic graph}
\newacro{DCE}{dead code elimination}
\newacro{DSL}{domain-specific language}
\newacro{GVN}{global value numbering}
\newacro{CPS}{continuation-passing style}
\newacro{LCSSA}{loop-closed SSA}
\newacro{JIT}{just-in-time}
\newacro{LoC}{lines of code}
\newacro{DOT}{dependent object types}
\newacro{NFA}{nondeterministic finite automaton}
\newacro{DFA}{deterministic finite automaton}
\newacro{ILP}{integer linear programming}
\newacro{SCC}{strongly connected component}
\newacro{SoN}{\enquote{Sea of Nodes}}

\DeclareAcronym{DTAL}{
    short = DTAL,
    long = Dependently Typed Assembly Language,
    cite = DBLP:conf/icfp/XiH01,
}

\DeclareAcronym{ArBB}{
    short = ArBB,
    long = Array Building Blocks,
    cite = DBLP:conf/cgo/NewburnSLMGTWDCWGLZ11,
}

\DeclareAcronym{CC}{
    short             = CC,
    long              = the Calculus of Constructions,
    cite              = DBLP:journals/iandc/CoquandH88,
}

\DeclareAcronym{CTRE}{
    short             = CTRE,
    long              = Compile Time Regular Expression,
    cite              = CTRE,
}

\DeclareAcronym{FFI}{
    short = FFI,
    long = foreign function interface,
    short-indefinite = an,
}

\DeclareAcronym{IR}{
    short             = IR,
    long              = intermediate representation,
    short-indefinite  = an,
}

\DeclareAcronym{PCRE2}{
    short             = PCRE2,
    long              = Perl-compatible Regular Expressions 2,
    cite              = pcre2,
}

\DeclareAcronym{PTS}{
    short             = PTS,
    short-plural-form = PTS,
    long              = pure type system,
    cite              = Bar:93,
}

\DeclareAcronym{RegEx}{
    short = RegEx,
    long = regular expression,
    short-plural = es,
}

\DeclareAcronym{SSA}{
    short             = {SSA},
    long              = {static single assignment},
    short-indefinite  = an,
}

\DeclareAcronym{SSI}{
    short             = SSI,
    long              = static single information,
    short-indefinite  = an,
}

\DeclareAcronym{SCCP}{
    short             = SCCP,
    long              = sparse conditional constant propagation,
    cite              = DBLP:conf/popl/WegmanZ85,
}

\DeclareAcronym{QTT}{
    short             = QTT,
    long              = Quantitative Type Theory,
    cite              = DBLP:conf/lics/Atkey18,
}

%% file: defs.tex
\def\Lcal{\mathcal{L}}

\def\lamG{\ensuremath{\lambda_G}\xspace}
\def\int{\hl{\mathtt{int}}}
\def\bool{\hl{\mathtt{bool}}}
\newcommand*{\var}[1][\ell]{\hl{\mathsf{var}\,#1}}
\newcommand*{\Var}[1]{\hl{\mathsf{var}\,{\mathtt{#1}}}}
\newcommand*{\val}[1]{\mathsf{val}\ #1}
\newcommand*{\Let}[3]{\hl{\mathtt{let}\ \mathtt{#1} = #2\mathtt{;}\ #3}}
\newcommand*{\hl}[1]{{\color{ACMBlue}#1}}

\newcommand*{\fv}[2][\hl{P}]{\mathit{fv}_{#1}(#2)}
\newcommand*{\FV}[2][\hl{P}]{\mathit{FV}_{#1}(#2)}
\newcommand*{\LV}[1]{\mathit{LV}(#1)}
\newcommand*{\LF}[1]{\mathit{LF}(#1)}
\newcommand*{\UF}[2][\hl{P}]{\mathit{UF}_{#1}(#2)}
\newcommand*{\dom}[1]{\mathit{dom}(#1)}
\def\unset{\!\!\uparrow}
\def\wf{\vdash_{\textnormal{wf}}}

\newcommand*{\inest}[2][\hl{P}]{\mathit{inest}_{#1}(#2)}

\NewDocumentCommand{\lam}{m O{\bot} m}{\ensuremath{\lambda #1 \to #2.\,#3}}

\newcommand{\irule}[3]{\mbox{\hypertarget{rule:#1}{\textsc{#1}}\inferrule{#2}{#3}}}

\newcommand{\iref}[1]{\hyperlink{rule:#1}{\textsc{#1}}}
\newcommand{\irefrule}[1]{\hyperlink{rule:#1}{Rule~\textsc{#1}}}

\usepackage{calc}
    
    \newcommand{\phantomr}[2]{\makebox[\widthof{#1}][r]{#2}}

\newcommand{\toolname}[1]{\textsc{#1}\@\xspace}
\newcommand{\mimir}{\toolname{MimIR}}
\newcommand{\llvm}{\toolname{LLVM}}
\newcommand{\mlir}{\toolname{MLIR}}
\newcommand{\ocaml}{\toolname{OCaml}}
\newcommand{\immer}{\toolname{immer}}
\newcommand{\lean}{\toolname{Lean}}

\newcommand{\thorin}{\toolname{Thorin}}
\newcommand{\guile}{\toolname{Guile}}

%% file: lst.tex
\usepackage{listings}
\usepackage{lstautogobble}
\newcommand{\tm}[1]{\tikz[remember picture,overlay]\coordinate (#1);}

\definecolor[named]{codegreen}    {named}{ACMGreen}
\definecolor[named]{codered}      {named}{ACMRed}
\definecolor[named]{codelightblue}{named}{ACMLightBlue}
\definecolor[named]{codedarkblue} {named}{ACMDarkBlue}
\definecolor[named]{whitesmoke}   {rgb}{0.96,0.96,0.96}

\def\lst{\lstinline}

\lstdefinelanguage{lamg}{
    literate=%
        {TO}{{$\mapsto$}}1
        {->}{{$\shortrightarrow$}}1
        {FORALL}{{$\color{ACMDarkBlue}\forall$}}1
        {BOT}{{$\color{ACMDarkBlue}\bot$}}1
        {LAM}{{${\color{ACMDarkBlue}\lambda}$}}1,
    morecomment = [s]{/*}{*/},
    morecomment = [l]{//},
    sensitive = true,
    keywords = [1]{int,bool,let},
    keywords = [2]{br},
}

\lstdefinelanguage{ssa}{
    literate=%
        {PHI}{{$\phi$}}1
        {->}{{$\shortrightarrow$}}1,
    morecomment = [s]{/*}{*/},
    morecomment = [l]{//},
    sensitive = true,
    keywords = [1]{int,bool},
    keywords = [2]{br,goto,return,fn,let},
}

%% file: abstract.tex
\begin{abstract}
    Dominance is a fundamental concept in compilers based on \ac{SSA} form.
    It underpins a wide range of analyses and transformations and defines a core property of \ac{SSA}:
    every use must be dominated by its definition.
    We argue that this reliance on dominance has become increasingly problematic---both in terms of precision and applicability to modern higher-order languages.
    First, control flow overapproximates data flow, which makes dominance-based analyses inherently imprecise.
    Second, dominance is well-defined only for first-order \acp{CFG}.
    More critically, higher-order programs violate the assumptions underlying \ac{SSA} and classic \acp{CFG}:
    without an explicit \ac{CFG}, the very notion that all uses of a variable must be dominated by its definition loses meaning.

    We propose an alternative foundation based on free variables.
    In this view, $\phi$-functions and function parameters directly express data dependencies, enabling analyses traditionally built on dominance while improving precision and naturally extending to higher-order programs.
    We further present an efficient technique for maintaining free-variable sets in a mutable \ac{IR}.
    For analyses requiring additional structure, we introduce the \emph{nesting tree}---a relaxed analogue of the dominator tree constructed from variable dependencies rather than control flow.

    Our benchmarks demonstrate that the algorithms and data structures presented in this paper scale log-linearly with program size in practice.
\end{abstract}

%% file: intro.tex
\section{Introduction}
\label{sec:intro}

Modern compilers rely on well-structured \acp{IR} to reason about programs.
Traditional \acp{IR} in \ac{SSA} form express control flow explicitly via a \ac{CFG} and use dominance to relate each use of a variable to its definition:
\begin{property}[Dominance]\label{prop:dom}
    Each use of a variable is dominated by its definition.
\end{property}
While this works well for first-order programs, dominance is inherently tied to \acp{CFG}.
Dominance becomes imprecise or even undefined for higher-order programs.
Yet higher-order functions are ubiquitous in modern code: in functional languages, in functional features of imperative languages like C++ and Java, and in data-parallel frameworks such as MapReduce or tensor computations.

\begin{wrapfigure}{r}{.3\textwidth}
\ifarxiv
\vspace{-.5ex}
\else
\vspace{1ex}
\fi
\begin{lstlisting}[basicstyle=\ttfamily\footnotesize,xleftmargin=1ex,xrightmargin=1ex,language=Caml]
let f x =
    let g y = y in
    g x
\end{lstlisting}
\vspace{-2ex}
\caption{\lst|g| does not depend on \lst|f|.}
\vspace{-2ex}
\label{fig:beta}
\end{wrapfigure}
$\lambda$-calculi handle higher-order programs naturally through block nesting:
variables are explicitly scoped by their syntactic structure.
However, this syntactic nesting can be too rigid for program transformations.
For example, $\beta$-reducing the application \lst|f a| in the \ocaml code in \cref{fig:beta} duplicates \lst|g|, although \lst|g| does not depend on \lst|f|'s variable \lst|x|.
For this reason, some \acp{IR}, such as \thorin~\cite{DBLP:conf/cgo/LeissaKH15}, \enquote{CPS soup} in \guile~\cite{wingolog2023approaching}, and MimIR~\cite{DBLP:journals/pacmpl/LeissaUMH25}, have abolished explicit scoping to simplify program transformations.
However, these approaches currently lack a formal metatheory and efficient, principled techniques for organizing a higher-order \enquote{soup of functions}, in particular for reconstructing explicit scope nesting.

This paper introduces \lamG, a graph-based $\lambda$-calculus for compiler \acp{IR}.
Like basic blocks in \ac{SSA}-based \acp{IR}, but unlike traditional $\lambda$-calculi, functions are \enquote{floating in a soup} and are \emph{not} explicitly nested.
Like traditional $\lambda$-calculi but unlike \ac{SSA}-based \acp{IR}, \lamG naturally supports higher-order functions.
Conceptually, \lamG replaces the CFG-based notion of dominance with a more relaxed relation of nesting:
instead of asking \enquote{Is $A$ dominated by $B$?}, we ask \enquote{Is $A$ nested within $B$?}.
This paper presents \lamG from two complementary perspectives—namely, as a \emph{mathematical calculus} and as a \emph{practical compiler IR} featuring efficient algorithms.

\subsection{Contributions}

In summary, this paper makes the following contributions:
\begin{itemize}[nosep,left=0pt .. \parindent]
    \item We present \lamG, a graph-based, typed $\lambda$-calculus that unifies functions and their bound variables through labels.
          It supports partially specified (unset) bodies and expresses (mutual) recursion through the graph structure rather than an explicit fixed-point operator.
          \lamG subsumes \acp{CFG} in \ac{SSA} form and naturally extends to higher-order programs (\cref{sec:overview}).

    \item We introduce a novel framework that \emph{lazily computes free variables on demand} and \emph{caches the results} for subsequent queries.
          Compiler passes can freely mutate the program representation and query free variables at any time.
          The cache is \emph{automatically and locally invalidated} when mutations occur, using a lightweight marking scheme that propagates invalidation only through affected parts of the program.
          This requires no manual reanalysis or synchronization with an external tool (\cref{sec:fvs}).

    \item We formally define a \emph{nesting relation} that specifies when a function must be nested within another.
          We prove the central insight of this work: nesting is a relaxed form of dominance.
          Unlike dominance, nesting depends solely on free variables rather than control flow, making it well-defined for higher-order programs.
          We restate the classic \ac{SSA}-dominance \cref{prop:dom} using nesting (\cref{sec:wf}).

    \item We propose a $\beta$-reduction algorithm for \lamG that specializes expressions only when necessary.
          Unlike conventional $\lambda$-calculi, our approach requires no block floating to hoist independent expressions out of functions.
          Furthermore, unlike scopeless representations such as \acp{CFG}, it avoids dominator trees and other control-flow preprocessing (\cref{sec:subst,sec:beta}).

    \item We generalize critical edge elimination for higher-order functions through $\eta$-expansion.
          While $\eta$-expansion itself is standard, to our knowledge this connection between the two techniques has not been explored before (\cref{sec:eta}).

    \item We introduce the \emph{nesting tree}, a relaxed alternative to the dominator tree that explicitly reconstructs the nesting structure of \lamG.
          Like nesting, it extends naturally to higher-order programs since its computation is independent of control flow.
          The nesting tree can optionally be enriched with information about loops and recursion (\cref{sec:nest_tree}).

    \item We present an immutable data structure supporting efficient persistent set operations, enabling scalable maintenance of free-variable sets (\cref{sec:trie}).

    \item We implemented the concepts presented in this paper in \mimir~\cite{DBLP:journals/pacmpl/LeissaUMH25}.
          Our benchmarks demonstrate that the proposed algorithms and data structures scale log-linearly with program size on realistic workloads (\cref{sec:eval}).
\end{itemize}
All lemmas and theorems in this paper have been verified with the proof assistant \lean.

%% file: overview.tex
\section{Overview}
\label{sec:overview}

This section first reviews classic \ac{SSA} and \acp{CFG}.
Then, we discuss \lamG and how to translate \iac{SSA}-based \ac{CFG} to~\lamG.
Finally, we present core ideas that are explored in more detail in later sections.

\subsection{CFG/SSA}

As an introductory example, consider the nested loops in the SSA-form \ac{CFG} in \cref{fig:loops:ssa}.
There is an outer loop with header \lst|hi|, exit block \lst|xi|, and body \lst|bi| that branches to the header \lst|hj| of the inner loop.
This inner loop, in turn, exits via \lst|xj| and has a body \lst|bj|.

Suppose we want to peel or unroll the outer loop.
To peel it,   we must specialize \lst|i1| to \lst|0|  (the first  argument of \lst|i1|'s $\phi$-function);
to unroll it, we must specialize \lst|i1| to \lst|j2| (the second argument of the $\phi$-function).
In both cases, specialization  must apply not only to \lst|hi| but also to all basic blocks that depend on \lst|i1|.
Compilers such as \llvm typically use the dominator tree of the \ac{CFG} (\cref{fig:loops:dom}) to identify these blocks, specializing all nodes dominated by the loop header.
In this example, all nodes except \lst|f| would be specialized.
However, note that the inner loop is independent of \lst|i1|.
Therefore, it is sufficient to specialize only \lst|hi|, \lst|bi|, and \lst|xi|.
This illustrates how control-flow information can overapproximate actual data dependencies (see also discussion of liveness in \cref{sec:relwork}).

\input{fig/syntax}
\input{fig/loops}
\input{fig/fv}

\subsection[From SSA to Lambda-G]{From SSA to \lamG}

We now translate this program into \lamG~(\cref{fig:loops:lam}).
A \lamG program~(\cref{fig:syntax}) consists of \emph{labels} mapped to \emph{functions}.
Each function is defined by a \emph{type} and a \emph{body}.
The body is either an \emph{expression} or \emph{unset}, denoted by the symbol~$\hl{\uparrow}$.
This enables partial or incremental graph construction.
For instance, to construct the cyclic program in \cref{fig:loops:lam}, we must \enquote{tie the knot}.
A straightforward approach is to first create all functions with unset bodies and then fill them in.
We also allow the program graph to be mutated---either by introducing new label-to-function mappings or by modifying the body of an existing function.

Each function introduces a variable identified by the same label as the function itself.
An expression may use $\hl{\ell}$ to refer to the function and $\var$ to refer to its variable.
The variable refers to the argument value bound upon application of the function, as in the standard $\lambda$‑calculus.
Since a program is a flat map from labels to functions, functions may reference variables defined in other functions \emph{without explicit lexical scoping}, as in a \ac{CFG}.
For example, \lst|bi| uses the variable \lst|i1 = $\Var{hj}$|.

\lamG's type system includes integers, booleans, tuple types, function types, and $\hl{\bot}$ (bottom)---a type with no inhabitants.
Functions whose codomain is~$\hl{\bot}$ are called \emph{continuations}.
Function application, tuple construction \& extraction, and \lst|let|-expressions---which allow for arbitrarily complex expressions---follow conventional syntax.
\lamG also includes constants for integers and booleans, as well as built-in functions, some of which may be written in infix form for readability.
In particular, we use a function \lst{br$_{T}$} of type \lst{[bool, [] -> T, [] -> T] -> T} for a conditional branch yielding a value of type $T$.
Thus, \lst{br$_{\bot}$(cond, t, f)} is like a conditional branch in \llvm or similar low-level \acp{IR}.

Using the correspondence between \ac{SSA} and \ac{CPS}~\cite{DBLP:journals/sigplan/Appel88,DBLP:conf/irep/Kelsey95}, we translate all basic blocks in \cref{fig:loops:ssa} with continuations in \cref{fig:loops:lam}.
$\phi$-Functions are represented as function variables, while their arguments become arguments to the application of the respective continuation, and the return point is made explicit:
instead of an implicit \lst|return|-terminator as in \ac{SSA}, the continuation \lst|ret| is passed to \lst|f|, which \lst|xi| invokes to exit \lst|f|.
This streamlining (treating basic blocks with $\phi$-functions as functions) allows us to implement loop peeling/unrolling as $\beta$-reduction (inlining):
We achieve peeling by $\beta$-reducing the call \lst|hi 0| and unrolling by $\beta$-reducing \lst|hi j2|.

\subsection{Nesting Tree}

As we have already argued, dominance only makes sense for first-order \acp{CFG}.
So how shall we identify which functions depend on \lst|hi| and need to be specialized as well?
For this reason, we present the \emph{nesting tree} depicted in \cref{fig:loops:nest} as an alternative to the dominator tree.
We say \enquote{$\ell_1$ nests $\ell_2$} if $\Var{\ell_1}$ occurs free in $\ell_2$ (or transitively).
Nesting implies dominance (\cref{thm:dom}) but not necessarily vice versa.
The nesting tree is not only well-defined for higher-order programs---unlike the dominator tree---but is also more precise since it is constructed from the free variables of the functions~(see \cref{fig:fvs} and \cref{sec:fvs}).
The nesting tree correctly captures that the inner loop is independent from the outer loop by placing them side by side.
Thus, $\beta$-reducing \lst|hi| will only specialize \lst|hi|, \lst|bi|, and \lst|xi|.
The nesting tree can be enhanced with \emph{sibling dependencies} (the dotted arrows in \cref{fig:loops:nest}) that track how functions at the same level depend on each other (see \cref{sec:siblings}).

\newcommand{\dep}{\mathrel{\tikz[baseline=-0.5ex]\draw[densely dotted,-{Latex[length=3pt]}](0,0)--(0.25,0);}}
Suppose we want to slightly change the program as indicated in the highlighted part in \cref{fig:loops:lam2}.
Note that applying this change to the SSA-form program in \cref{fig:loops:ssa} does \emph{not} change the dominator tree in \cref{fig:loops:dom}.
However, the nesting tree \emph{does} change as the inner loop now depends on the outer loop (see \cref{fig:loops:nest2}).
The nesting tree closely resembles the dominator tree, except that \lst|hj| hangs directly under \lst|hi|.
The reason for this is that \lst|bi| does not introduce any variables that \lst|hj| depends on.

\subsection{Higher-Order Functions and Direct Style}

\input{fig/concl}
\input{fig/iter}

\lamG not only subsumes classical \acp{CFG} in \ac{SSA} form, but also represents higher‑order functions in both \ac{CPS} and direct style.
\Cref{fig:iter:iter} depicts a \lamG program that computes $f^n(x)$ in direct style.
\Cref{fig:iter:pow} successively builds upon \lst|iter| to construct a \mbox{\lst|pow|er} function via currying and partial application.
For example, \lst|pow 3 5| evaluates to \lst|243|.
\Cref{fig:iter:nest} presents the resulting nesting forest of the program.

This example illustrates a qualitative advantage of \lamG over \ac{CFG}-based \acp{IR}:
\lamG can directly express higher-order functions like \lst|iter|.
\acp{CFG} cannot directly represent this program without further lowering.
The nesting forest automatically recovers the dependency structure from free variables alone---without any explicit lexical scoping.

\subsection{Summary}

While the nesting tree is useful for making dependencies between functions explicit, \emph{an explicit construction is rarely necessary}.
In most cases, free-variable information suffices, as in the dependency-aware substitution described in \cref{sec:subst}.

This paper explores the relationship between \ac{SSA} form and functional programming, and discusses further generalizations (summarized in \cref{tbl:concl}) connecting \ac{SSA} and \lamG.
This relationship is not a strict mathematical isomorphism, but rather illustrates how key concepts in \ac{SSA} correspond to those in \lamG.
In particular, while free variables correspond to liveness in \ac{SSA} (see \cref{sec:relwork}), in \cref{tbl:concl} they serve as the primary structural notion in place of dominance.

%% file: fig/syntax.tex
\begin{figure}[p]
	\footnotesize
	\begin{align*}
		\hl{P} & : \hl{\Lcal}  \to \hl{f}                                                                                             &  & \text{program: maps labels to functions}                 \\
		\hl{f} & \coloneqq \hl{\lam{t}[t]{b}}                                                                                         &  & \text{function}                                          \\
		\hl{b} & \coloneqq \hl{e} \mid \hl{\uparrow}                                                                                  &  & \text{body: expr/unset}                                  \\
		\hl{e} & \coloneqq \hl{c} \mid \hl{\ell} \mid \var{} \mid \hl{e\ e} \mid \hl{(e, \ldots, e)} \mid \hl{e.i} \mid \Let{x}{e}{e} \mid \hl{\mathtt{x}} &  & \text{expr: constant/function/var/app/tuple/extract/let} \\
		\hl{t} & \coloneqq \int \mid \bool \mid \hl{\bot} \mid \hl{t \to t} \mid \hl{[t, \ldots, t]}                                  &  & \text{type: int/bool/bottom/function type/tuple type}
	\end{align*}
	\vspace{-5ex}
	\caption{Syntax of \lamG}
	\vspace{-2ex}
	\label{fig:syntax}
\end{figure}

%% file: fig/loops.tex
\begin{figure}[p]
	\begin{subcaptionblock}{.19\textwidth}
		\begin{lstlisting}[basicstyle=\ttfamily\tiny,xleftmargin=1ex,xrightmargin=1ex,language=ssa,framerule=0pt,belowskip=2.5pt]
            $\tm{fstart}$int f(int n)

              goto hi         $\tm{fend}$

            $\tm{histart}$hi:
              i1 = PHI(0, j2)
              br(i1<n, bi, xi)$\tm{hiend}$

            $\tm{bistart}$bi:
              goto hj         $\tm{biend}$

            $\tm{hjstart}$hj:
              j1 = PHI(i1, j2)
              j2 = j1 + 1
              br(j1<n, bj, xj)$\tm{hjend}$

            $\tm{bjstart}$bj:
              goto hj         $\tm{bjend}$

            $\tm{xjstart}$xj:

              goto hi         $\tm{xjend}$

            $\tm{xistart}$xi:
              return i1       $\tm{xiend}$
        \end{lstlisting}
		\vspace{-1ex}
		\caption{SSA/CFG}
		\label{fig:loops:ssa}
	\end{subcaptionblock}
	\begin{subcaptionblock}{.245\textwidth}
		\begin{lstlisting}[basicstyle=\ttfamily\tiny,xleftmargin=1ex,xrightmargin=1ex,language=lamg]
            f TO LAM[int, int -> BOT] -> BOT.
              let n   = $\Var{f}$.0;
              let ret = $\Var{f}$.1;
              hi 0

            hi TO LAM int -> BOT.
              let i1 = $\Var{hi}$;
              br$_\bot$(i1<n, bi, xi)

            bi TO LAM[] -> BOT.
              hj i1

            hj TO LAM int -> BOT.
              let j1 = $\Var{hj}$;
              let j2 = j1 + 1;
              br$_\bot$(j1<n, bj, xj)

            bj TO LAM[] -> BOT.
              hj j2

            xj TO LAM[] -> BOT.

              hi j2

            xi TO LAM[] -> BOT.
              ret i1
        \end{lstlisting}
		\vspace{-1ex}
		\caption{\lamG}
		\label{fig:loops:lam}
	\end{subcaptionblock}
	\begin{subcaptionblock}{.245\textwidth}
		\begin{lstlisting}[basicstyle=\ttfamily\tiny,xleftmargin=1ex,xrightmargin=1ex,language=lamg]
            f TO LAM[int, int -> BOT] -> BOT.
              let n   = $\Var{f}$.0;
              let ret = $\Var{f}$.1;
              hi 0

            hi TO LAM int -> BOT.
              let i1 = $\Var{hi}$;
              br$_\bot$(i1<n, bi, xi)

            bi TO LAM[] -> BOT.
              hj i1

            hj TO LAM int -> BOT.
              let j1 = $\Var{hj}$;
              let j2 = j1 + 1;
              br$_\bot$(j1<n, bj, xj)

            bj TO LAM[] -> BOT.
              hj j2

            xj TO LAM[] -> BOT.
              $\color{ACMRed}\texttt{\textbf{let} i2 = i1 + j1;}$
              hi $\color{ACMRed}\texttt{i2}$

            xi TO LAM[] -> BOT.
              ret i1
        \end{lstlisting}
		\vspace{-1ex}
		\caption{\lamG w/ \texttt{i2}}
		\label{fig:loops:lam2}
	\end{subcaptionblock}
	\begin{subcaptionblock}{.25\textwidth}
		\centering
		\begin{tikzpicture}
			\node[n] (f)                      { \texttt{f}  };
			\node[n] (hi) [below      =of  f] { \texttt{hi} };
			\node[n] (bi) [below  left=of hi] { \texttt{bi} };
			\node[n] (xi) [below right=of hi] { \texttt{xi} };
			\node[n] (hj) [below      =of bi] { \texttt{hj} };
			\node[n] (bj) [below  left=of hj] { \texttt{bj} };
			\node[n] (xj) [below right=of hj] { \texttt{xj} };

			\draw (f)   -- (hi);
			\draw (hi)  -- (bi);
			\draw (hi)  -- (xi);
			\draw (bi)  -- (hj);
			\draw (hj)  -- (bj);
			\draw (hj)  -- (xj);
		\end{tikzpicture}
		\vspace{-1ex}
		\caption{Dominator tree of \ref{fig:loops:ssa}}
		\label{fig:loops:dom}
		\begin{tikzpicture}
			\node[n] (f) { \texttt{f} };
			\node[n] (hi) [below  left=of f] { \texttt{hi} };
			\node[n] (hj) [below right=of f] { \texttt{hj} };
			\node[n] (xi) [below=of hi] { \texttt{xi} };
			\node[n] (xj) [below=of hj] { \texttt{xj} };
			\node[n] (bi) [ left=of xi] { \texttt{bi} };
			\node[n] (bj) [right=of xj] { \texttt{bj} };

			\draw (f)  -- (hi);
			\draw (f)  -- (hj);
			\draw (hi)  -- (bi);
			\draw (hi)  -- (xi);
			\draw (hj)  -- (bj);
			\draw (hj)  -- (xj);

			\draw[a] (hj) to[loop right] (hj);
			\draw[a] (hj) to[bend right] (hi);
			\draw[a] (hi) to[bend right] (hj);
		\end{tikzpicture}
		\caption{Nesting tree of \ref{fig:loops:lam}}
		\label{fig:loops:nest}
		\begin{tikzpicture}
			\node[n] (f)                      { \texttt{f}  };
			\node[n] (hi) [below      =of  f] { \texttt{hi} };
			\node[n] (bi) [below      =of hi] { \texttt{bi} };
			\node[n] (xi) [right      =of bi] { \texttt{xi} };
			\node[n] (hj) [ left      =of bi] { \texttt{hj} };
			\node[n] (bj) [below  left=of hj] { \texttt{bj} };
			\node[n] (xj) [below right=of hj] { \texttt{xj} };

			\draw (f)   -- (hi);
			\draw (hi)  -- (bi);
			\draw (hi)  -- (xi);
			\draw (hi)  -- (hj);
			\draw (hj)  -- (bj);
			\draw (hj)  -- (xj);

			\draw[a] (hi) to[loop right] (hi);
			\draw[a] (hj) to[loop  left] (hj);
			\draw[a] (bi) to (hj);
		\end{tikzpicture}
		\caption{Nesting tree of \ref{fig:loops:lam2}}
		\label{fig:loops:nest2}
	\end{subcaptionblock}

    \vspace{1ex}
	\begin{subcaptionblock}{.48\textwidth}
        \vspace{1ex}\centering
\begin{lstlisting}[language=Caml,basicstyle=\ttfamily\tiny,morekeywords=int]
let f ((n: int), (ret: int -> 'bot)): 'bot =
  let rec hi i1 =
    let bi () = hj  i1 in
    let xi () = ret i1 in
    (if i1<n then bi else xi) ()
  and hj(j1: int) =
    let j2 = j1 + 1 in
    let bj() = hj j2 in
    let xj() = hi j2 in
    (if j1<n then bj else xj) ()
  in hi 0
\end{lstlisting}
		\vspace{-1ex}
		\caption{\cref{fig:loops:ssa}/\subref{fig:loops:lam} in \ocaml}
		\label{fig:loops:ml}
	\end{subcaptionblock}
    \hspace{1ex}
	\begin{subcaptionblock}{.48\textwidth}
        \vspace{1ex}
		\centering
\begin{lstlisting}[language=Caml,basicstyle=\ttfamily\tiny,morekeywords=int]
let f ((n: int), (ret: int -> 'bot)): 'bot =
  let rec hi i1 =
    let rec hj(j1: int) =
      let j2 = j1 + 1 in
      let bj() = hj j2 in
      let xj() = let i2 = i1 + j1 in hi i2 in
      (if j1<n then bj else xj) () in
    let bi () = hj  i1 in
    let xi () = ret i1 in
    (if i1<n then bi else xi) ()
  in hi 0
\end{lstlisting}
		\vspace{-1ex}
		\caption{\cref{fig:loops:lam2} in \ocaml}
		\label{fig:loops:ml2}
	\end{subcaptionblock}
	\begin{tikzpicture}[remember picture, overlay,
		block/.style={shape=rectangle, draw, semithick, inner xsep=0pt, inner ysep=3.2pt, yshift=1.7pt, fill=none, fit={(#1start.north east) (#1end.south west)}},
		arrow/.style={-{Latex[length=3pt]}, semithick}
		]
		\node[block=f]  (F) {};
		\node[block=hi] (HI) {};
		\node[block=hj] (HJ) {};
		\node[block=bi] (BI) {};
		\node[block=bj] (BJ) {};
		\node[block=xj] (XJ) {};
		\node[block=xi] (XI) {};

		\draw[arrow] (F.south)  -- (HI.north);
		\draw[arrow] (HI.south) -- (BI.north);
		\draw[arrow] (BI.south) -- (HJ.north);
		\draw[arrow] (HJ.south) -- (BJ.north);
		\draw[arrow] (HI.west) -- ++(-8pt,0pt) |- (XI.west);
		\draw[arrow] (HJ.west) -- ++(-4pt,0pt) |- (XJ.west);
		\draw[arrow] (BJ.east) -- ++(+4pt,0pt) |- (HJ.east);
		\draw[arrow] (XJ.east) -- ++(+8pt,0pt) |- (HI.east);
	\end{tikzpicture}
	\vspace{-2ex}
	\caption{Two nested loops. Inner loop (\texttt{j1}/\texttt{j2}) only depends on outer loop (\texttt{i1}) in \ref{fig:loops:lam2} by \texttt{i2}.}
	\label{fig:loops}
\end{figure}

%% file: fig/fv.tex
\begin{table}[p]
	\caption{Computations and invalidation of free variables (\lst|m| = \lst|mark|)}
	\label{fig:fvs}\vspace{-2ex}
	\resizebox{\textwidth}{!}{%
		{\tiny\setlength{\tabcolsep}{1pt}%
				\begin{tabular}{crrcrcrcrcrcrcrc} \toprule
					 & \multirow{3}{*}{\rotatebox{90}{\lst|run|}} & \multicolumn{2}{c}{\texttt{f}} & \multicolumn{2}{c}{\texttt{hi}} & \multicolumn{2}{c}{\texttt{bi}} & \multicolumn{2}{c}{\texttt{hj}} & \multicolumn{2}{c}{\texttt{bj}} & \multicolumn{2}{c}{\texttt{xj}} & \multicolumn{2}{c}{\texttt{xi}}                                                                                                                                                           \\\cmidrule(lr){3-4}\cmidrule(lr){5-6}\cmidrule(lr){7-8}\cmidrule(lr){9-10}\cmidrule(lr){11-12}\cmidrule(lr){13-14}\cmidrule(lr){15-16}
					 &                                            & \lst|m|                        & $\mathit{FV}$                   & \lst|m|                         & $\mathit{FV}$                   & \lst|m|                         & $\mathit{FV}$                   & \lst|m|                         & $\mathit{FV}$           & \lst|m| & $\mathit{FV}$                     & \lst|m| & $\mathit{FV}$                     & \lst|m| & $\mathit{FV}$           \\\midrule
					\multirow{5}{*}{\rotatebox{90}{\cref{fig:loops:lam}}}
					 & 0                                          & 0                              & $\varnothing$                   & 0                               & $\varnothing$                   & 0                               & $\varnothing$                   & 0                               & $\varnothing$           & 0       & $\varnothing$                     & 0       & $\varnothing$                     & 0       & $\varnothing$           \\
					 & 2                                          & 0                              & $\varnothing$                   & 0                               & $\varnothing$                   & 0                               & $\varnothing$                   & 0                               & $\varnothing$           & 0       & $\varnothing$                     & 0       & $\varnothing$                     & 2       & $\{\Var{f}, \Var{hi}\}$ \\\cmidrule(lr){2-16}
					 & 4                                          & 4                              & $\varnothing$                   & 4                               & $\{\Var{f}\}$                   & 4                               & $\{\Var{f}, \Var{hi}\}$         & 4                               & $\{\Var{f}          \}$ & 4       & $\{\Var{hj}          \}$          & 4       & $\{\Var{hj}          \}$          & 2       & $\{\Var{f}, \Var{hi}\}$ \\
					 & 5                                          & 5                              & $\varnothing$                   & 5                               & $\{\Var{f}\}$                   & 5                               & $\{\Var{f}, \Var{hi}\}$         & 5                               & $\{\Var{f}          \}$ & 5       & $\{\Var{f}, \Var{hj}          \}$ & 5       & $\{\Var{f}, \Var{hj}          \}$ & 2       & $\{\Var{f}, \Var{hi}\}$ \\
					 & 6                                          & 6                              & $\varnothing$                   & 6                               & $\{\Var{f}\}$                   & 6                               & $\{\Var{f}, \Var{hi}\}$         & 6                               & $\{\Var{f}          \}$ & 6       & $\{\Var{f}, \Var{hj}          \}$ & 6       & $\{\Var{f}, \Var{hj}          \}$ & 2       & $\{\Var{f}, \Var{hi}\}$ \\\midrule
					 & 6                                          & 0                              & $\varnothing$                   & 0                               & $\varnothing$                   & 0                               & $\varnothing$                   & 0                               & $\varnothing$           & 0       & $\varnothing$                     & 0       & $\varnothing$                     & 2       & $\{\Var{f}, \Var{hi}\}$ \\\midrule
					\multirow{3}{*}{\rotatebox{90}{\cref{fig:loops:lam2}}}
					 & 8                                          & 8                              & $\varnothing$                   & 8                               & $\{\Var{f}\}$                   & 8                               & $\{\Var{f}, \Var{hi}\}$         & 8                               & $\{\Var{f}, \Var{hi}\}$ & 8       & $\{\Var{hj}\}$                    & 8       & $\{\Var{hj}, \Var{hi}\}$          & 2       & $\{\Var{f}, \Var{hi}\}$ \\
					 & 9                                          & 9                              & $\varnothing$                   & 9                               & $\{\Var{f}\}$                   & 9                               & $\{\Var{f}, \Var{hi}\}$         & 9                               & $\{\Var{f}, \Var{hi}\}$ & 9       & $\{\Var{f}, \Var{hi}, \Var{hj}\}$ & 9       & $\{\Var{f}, \Var{hi}, \Var{hj}\}$ & 2       & $\{\Var{f}, \Var{hi}\}$ \\
					 & 10                                         & 10                             & $\varnothing$                   & 10                              & $\{\Var{f}\}$                   & 10                              & $\{\Var{f}, \Var{hi}\}$         & 10                              & $\{\Var{f}, \Var{hi}\}$ & 10      & $\{\Var{f}, \Var{hi}, \Var{hj}\}$ & 10      & $\{\Var{f}, \Var{hi}, \Var{hj}\}$ & 2       & $\{\Var{f}, \Var{hi}\}$ \\\bottomrule
				\end{tabular}%
			}%
	}
\end{table}

%% file: fig/concl.tex
\colorlet{group}{gray!25}%
\colorlet{group2}{gray!50}%
\begin{table}[t]
	\caption{SSA/CFG concepts and their analogues in \lamG}\vspace{-2ex}%
	\label{tbl:concl}%
    \footnotesize\centering\hspace{1ex}%
    \begin{minipage}[t]{0.45\textwidth}%
	\begin{tabular}{ll}
		\toprule
		CFG w/ SSA                   & $\lamG$                                      \\
		\midrule
        \rowcolor{group}Function    &                                              \\
		\rowcolor{group}Basic block & \multirow{-2}{*}{\cellcolor{group}Function}  \\
		Function parameter           & \multirow{2}{*}{$\lambda$-Variable}          \\
		$\phi$-Function                                                             \\
		\rowcolor{group}Function call &                                            \\
		\rowcolor{group}Jump  & \multirow{-2}{*}{\cellcolor{group}Application}     \\
		Function call argument      & \multirow{2}{*}{Application argument}          \\
		$\phi$-Argument                                                              \\
		\rowcolor{group}Function inlining           &              \\
		\rowcolor{group}Loop peeling                &                                                \\
		\rowcolor{group}Loop unrolling              &                                                \\
		\rowcolor{group}CFG specialization          & \multirow{-4}{*}{\cellcolor{group}$\beta$-Reduction}      \\
		\bottomrule
	\end{tabular}
    \end{minipage}\hfill%
    \begin{minipage}[t]{0.5\textwidth}
	\begin{tabular}{ll}
		\toprule
		CFG w/ SSA                  & $\lamG$                                        \\\midrule
		\rowcolor{group}SSA dominance property      & Well-Formedness                                \\
		Dominance                   & Free variables                                  \\
		\rowcolor{group}Dominator tree              & Nesting tree                                   \\
		\rowcolor{group2}DJ-Graph                   & \\
		\rowcolor{group2}Call graph                 & \multirow{-2}{*}{\cellcolor{group2}\quad w/ sibling dependencies}\\
		\rowcolor{group}Loop Tree                   &                                                  \\
		\rowcolor{group}Call graph SCCs             & \multirow{-2}{*}{\cellcolor{group}\quad\quad w/ SCCs}             \\
		Directly recursive function & \multirow{2}{*}{Direct recursion   }           \\
		Natural loop                                                                 \\
		\rowcolor{group}Mutually recursive function  & \\
		\rowcolor{group}Irreducible loop            & \multirow{-2}{*}{\cellcolor{group}Mutual recursion}              \\
		Critical edge elimination   & $\eta$-Expansion                               \\
		\bottomrule
	\end{tabular}
    \end{minipage}\hspace{1ex}
\end{table}

%% file: fig/iter.tex
\begin{figure}[t]
	\begin{subcaptionblock}{.34\textwidth}
		\begin{lstlisting}[basicstyle=\ttfamily\tiny,xleftmargin=1ex,xrightmargin=1ex,language=lamg]
            iter TO LAM[int -> int, int, int] -> int.
                let f = $\Var{iter}$.0;
                let n = $\Var{iter}$.1;
                let x = $\Var{iter}$.2;
                br$_\hl{\mathtt{int}}$(n <= 0, a, b)
            a TO LAM[] -> int. x
            b TO LAM[] -> int. iter (f, n - 1, f x)
        \end{lstlisting}
        \vspace{-1ex}
		\caption{\lst|iter| computes $f^n(x)$}
		\label{fig:iter:iter}
	\end{subcaptionblock}
	\begin{subcaptionblock}{.50\textwidth}
		\begin{lstlisting}[basicstyle=\ttfamily\tiny,xleftmargin=1ex,xrightmargin=1ex,language=lamg]
            succ TO LAM int -> int.       $\Var{succ}$ + 1
            add  TO LAM int -> int -> int. add'             // curried
            add' TO LAM int -> int.       iter (succ, $\Var{add}$, $\Var{add'}$)
            mul  TO LAM int -> int -> int. mul'             // curried
            mul' TO LAM int -> int.       iter (add $\Var{mul}$, $\Var{mul'}$, 0)
            pow  TO LAM int -> int -> int. pow'             // curried
            pow' TO LAM int -> int.       iter (mul $\Var{pow}$, $\Var{pow'}$, 1)
        \end{lstlisting}
        \vspace{-1ex}
		\caption{Successive use of \lst|iter| to build \lst|pow|er}
		\label{fig:iter:pow}
	\end{subcaptionblock}
	\begin{subcaptionblock}{.14\textwidth}
		\begin{tikzpicture}
			\node[n] (f)                    { \texttt{iter} };
			\node[n] (a) [below  left=of f] { \texttt{a} };
			\node[n] (b) [below right=of f] { \texttt{b} };
			\node[n]     [right =.25cm of f]      { \texttt{succ} };

			\draw (f)  -- (a);
			\draw (f)  -- (b);
		\end{tikzpicture}

        \vspace{.25cm}

		\begin{tikzpicture}
			\node[n] (add)                  { \texttt{add}  };
			\node[n] (addc) [below=of add]  { \texttt{add'} };
			\node[n] (mul)  [right=of add]  { \texttt{mul}  };
			\node[n] (mulc) [below=of mul]  { \texttt{mul'} };
			\node[n] (pow)  [right=of mul]  { \texttt{pow}  };
			\node[n] (powc) [below=of pow]  { \texttt{pow'} };

			\draw (add) -- (addc);
			\draw (mul) -- (mulc);
			\draw (pow) -- (powc);
		\end{tikzpicture}
        \vspace{-1ex}
        \caption{Nest. forest}
        \label{fig:iter:nest}
	\end{subcaptionblock}
    \vspace{-2ex}
    \caption{Higher-order function \lst|iter| to construct the \lst|pow|er function}
    \label{fig:iter}
\end{figure}

%% file: sema.tex
\section{Semantics}
\label{sec:sema}

This section studies $\lamG$.
Its most intricate feature is its scopeless graph structure, reminiscent of a \ac{CFG}.
We restrict our presentation to the rules for variables, functions, and applications, which capture the essential behavior of \lamG.
For brevity, we omit rules for tuples and extractions as they are similar to applications, and we skip rules for constants since they are trivial.
While \lamG includes \lst|let|-bindings in the syntax, we omit them from the semantics.
Adding them is straightforward but would introduce unnecessary notation without contributing new ideas.
In practice, our implementation operates on a \ac{SoN} (see \cref{sec:relwork}), making explicit \lst|let|-bindings unnecessary;
they are used only in examples to provide a readable, textual representation of these graphs.
One can think of these examples as if all \lst|let|-bound variables had been inlined with their definitions, with the underlying graph representation sharing common (sub-)expressions as an implementation detail.

\paragraph{Typing}

Since a \lamG program already consists of a label-to-function map~\hl{P}, we do not need an additional typing environment for type checking; instead, we look up labels directly in~$\hl{P}$ (\irefrule{T-Fun}/\iref{T-Var}).
Moreover, because function types are fully annotated, function bodies need not be type-checked at use sites (\irefrule{T-Fun}).
This design allows the \ac{IR} to compute and store the type of each expression at construction time.
A function is type-checked when its body is defined (\irefrule{T-Body}), and the entire program can be type-checked (\irefrule{T-Prog}) by verifying all function bodies and their subexpressions.
The typing rules themselves are standard:
\begin{mathpar}
	\irule{T-App}{
		\hl{P} \vdash \hl{e_1} : \hl{t \to u} \\
		\hl{P} \vdash \hl{e_2} : \hl{t}
	}{
		\hl{P} \vdash \hl{e_1\ e_2} :  \hl{u}
	}
    \and
	\irule{T-Fun}{
        \hl{P}(\hl{\ell}) = \hl{\lam{t}[u]{e}} \\
	}{
		\hl{P} \vdash \hl{\ell} :  \hl{t \to u}
	}
    \and
	\irule{T-Var}{
        \hl{P}(\hl{\ell}) = \hl{\lam{t}[u]{e}} \\
	}{
		\hl{P} \vdash \hl{\var} :  \hl{t}
	}
    \and
	\irule{T-Body}{
		\hl{P} \vdash \hl{e} : \hl{u}
	}{
		\hl{P} \vdash \hl{\lam{t}[u]{e}}
	}
    \and
	\irule{T-Prog}{
		\hl{P} \vdash \hl{\lam{t_1}[u_1]{b_1}} \\
        \cdots \\
		\hl{P} \vdash \hl{\lam{t_n}[u_n]{b_n}}
	}{
		\vdash \underbrace{\{\hl{\ell_1} \mapsto \hl{\lam{t_1}[u_1]{b_1}}, \ldots, \hl{\ell_n} \mapsto \hl{\lam{t_n}[u_n]{b_n}}\}}_\hl{P}
	}
\end{mathpar}

\subsection{Free Variables}
\label{sec:fvs}

We call the variables that occur in an expression up to (but not including) function expressions the \emph{local variables} of the expression:
\begin{align}
	\LV{\hl{\var{}}}   & = \{\hl{\var{}}\}                                                                                                  \\
    \label{eq:lv_app}
	\LV{\hl{e_1\ e_2}} & = \LV{\hl{e_1}} \cup \LV{\hl{e_2}}                                                                                 \\
	\LV{\hl{\ell}}   & = \varnothing                                                                                                      \\
	\intertext{Similarly, we call the functions that occur in an expression up to (but not descending into) function expressions the \emph{local functions} of the expression:}
    \label{eq:lf1}
	\LF{\hl{\var{}}}   & = \varnothing                                                                                                      \\
	\LF{\hl{e_1\ e_2}} & = \LF{\hl{e_1}} \cup \LF{\hl{e_2}}                                                                                 \\
    \label{eq:lf3}
	\LF{\hl{\ell}}   & = \{ \hl{\ell} \}
	\intertext{This allows us to define \emph{free variables} as the solution to the following recursive equations:}
    \label{eq:fv1}
	\fv{\hl{\var{}}}   & = \{\hl{\var{}}\}                                                                                                  \\
    \label{eq:fv2}
	\fv{\hl{e_1\ e_2}} & = \LV{\hl{e_1\ e_2}} \cup \fv{\LF{\hl{e_1\ e_2}}}                                                                  \\
    \label{eq:fv3}
	\fv{\hl{\ell}}     & = \varnothing                                     &  & \text{if $\hl{P}(\hl{\ell}) = \hl{\lam{t}[u]{\unset}}$} \\
    \label{eq:fv_lam}
	\fv{\hl{\ell}}     & = \fv{\hl{e}} \setminus \{\hl{\var{}}\}           &  & \text{if $\hl{P}(\hl{\ell}) = \hl{\lam{t}[u]{\hl{e}}}$} \\
	\intertext{We obtain the final free-variable information as the least fixed point of the recurrence above:}
	\FV{\hl{e}}        & = \mu\,\fv{\hl{e}}
\end{align}

\subsubsection{Implementation}
\label{sec:fv_impl}

Our implementation allows for efficient in-place mutation of function bodies, such as changing the body of \lst|xj| from \cref{fig:loops:lam} to \ref{fig:loops:lam2}.
Expressions, however, are immutable: once created, they cannot be modified.
This allows us to store the local variables and functions during the construction of an expression, as this information remains sound regardless of subsequent mutations to the \ac{IR}.
Thus, local variables and functions can be obtained without descending into subexpressions.

\paragraph{Fixed-Point Iteration}

If the compiler engineer queries the free variables of a function, we solve the fixed point and memoize the result so that subsequent queries return the memoized set directly.
Expressions compute their free variables efficiently via \cref{eq:fv2}, using local variables and the (memoized) free-variable sets of the local functions.
To efficiently track which parts of the program require free-variable (re-)computation, and to prevent infinite cycles during the fixed-point iteration, we associate an integer \verb|mark| with each function.
Initially, \verb|mark| is set to \lst|0|, indicating that the free-variable set is invalid.
The fixed-point iteration increments a counter \verb|run| in each iteration to distinguish the following cases:
\begin{enumerate}[nosep,left=0pt .. \parindent,leftmargin=1.5em]
	\item $\mathtt{mark} = 0$: The free-variable set is invalid; set \verb|mark| to \verb|run| and recompute it.
	\item $\mathtt{mark} = \mathtt{run} - 1$: The free-variable set stems from the previous iteration.
	      Set \verb|mark| to \verb|run| and add any additional variables discovered in this iteration.
	\item $\mathtt{mark} = \mathtt{run}$: The expression belongs to the current iteration, indicating a cyclic dependency.
	      In this case, we return the free-variable set accumulated so far.
	\item Otherwise: The free-variable set is valid. Yield the memoized set without further computation.
\end{enumerate}
To distinguish cases~2 and~3, each new fixed-point iteration increments \verb|run| by~2.
After the fixed-point loop has terminated, no further traversals are required.
The value of \verb|mark| already indicates that the memoized free-variable set is valid (case~4).

To solve the data-flow equations, we require at most $d(G) + 2$ iterations, where $d(G)$ is the loop-connectedness of $G$.
This is the maximum number of back edges on any acyclic path in~$G$~\cite{DBLP:journals/jacm/KamU76,cooper2004iterative}.

A fixed-point loop requires two traversals for acyclic programs ($d(G) = 0$):
the first computes the initial solution, and the second verifies that the fixed point has been reached.
As described above, the fixed-point loop detects cyclic dependencies (case~3).
This enables an additional optimization when no cycles occur.
In this case, computation terminates immediately after the first iteration, as the computed free-variable sets are already sound.

\paragraph{Users \& Invalidation}

The memoized information remains valid as long as no referenced function is modified.
For this reason, each function maintains a set of functions that \emph{use} it:
\begin{equation}
	\UF{\hl{\ell}} = \{\hl{\ell_\mathit{use}} \mid \hl{\ell_\mathit{use}} \to \hl{\ell}\} \qquad\text{where}\qquad
    	\irule{Succ}{
		\hl{P}(\hl{\ell_1}) = \hl{\lam{t}[u]{e}} \\
		\hl{\ell_2} \in \LF{\hl{e}}
	}{
		\hl{\ell_1} \to_\hl{P} \hl{\ell_2}
    }
\end{equation}
If a function is modified, we transitively \emph{invalidate} all functions in its \emph{users} set by setting their \verb|mark| to \lst|0|.
This marks the memoized free-variable sets as invalid, requiring recomputation upon the next query.
This requires traversing only a local part of the program:
first, $\fv{\hl{e}}$ processes only expressions reachable from $\hl{e}$;
second, expressions with valid free-variable sets require no further traversal;
finally, computation is lazy and performed only on demand.

This user set is the \emph{only} reverse dependency we require.
In particular, we do not track classic def-use chains.

\paragraph{Example}

Consider \cref{fig:loops:lam} and \cref{fig:fvs}.
Initially, \lst|mark| is $0$, indicating that the free-variable cache (column $\mathit{FV}$) is invalid.
Suppose we wish to compute the free variables of \lst|xi|.
These are given by the local variables of its body, united with the free variables of the local functions appearing in that body while removing the variable introduced in \lst|xi| itself:
\begin{equation*}
    \FV{\mathtt{xi}}
    = \LV{\mathtt{ret}\ \mathtt{i1}} \cup \FV{\LF{\mathtt{ret}\ \mathtt{i1}}} \setminus \{\Var{xi}\}
    = \{\Var{f}, \Var{hi}\}
\end{equation*}
Local variables and local functions of every expression are precomputed at construction time, so no descent into subexpressions is required.
Since this is a new computation, we increase \lst|run| by two.
Because the computation does not involve cycles, the result is already sound after a single iteration, and no further traversal is necessary.

Next, consider computing the free variables of \lst|f|.
We again increase \lst|run| by two, to $4$, and recursively compute free variables, marking each visited function with this value.
Starting from \lst|f|, we reach \lst|hi|, then \lst|bi|, \lst|hj|, \lst|bj|, and finally revisit \lst|hj|, as indicated by $\mathtt{mark} = \mathtt{run} = 4$.
This signals a cyclic dependency, and we therefore use the partial free-variable set for \lst|hj| accumulated so far.
Continuing with the computation for \lst|xj| requires revisiting \lst|hi|, which is likewise marked with the current value of \lst|run|; consequently, we again return the partial result computed so far.
Finally, we reach \lst|xi|, whose \lst|mark| indicates that its cached free-variable set is already sound, so we return the memoized set directly.
This completes the first iteration of the fixed-point computation.
In the second iteration ($\mathtt{run} = 5$), we propagate $\Var{f}$ to \lst|bj| and \lst|xj|.
The computation stabilizes in the third iteration ($\mathtt{run} = 6$), after which all computed $\mathit{FV}$ sets are sound.
Subsequent free-variable queries for any of these functions return the memoized results directly.

Now suppose we unset \lst|xj|, which invalidates all functions whose free-variable information depends on \lst|xj|, leaving only \lst|xi| valid.
We then set the body of \lst|xj| to \lst|hi i2| (as in \cref{fig:loops:lam2}) and again query the free variables of \lst|f|.
We increase \lst|run| by two, to $8$, and recursively compute free variables while updating \lst|mark| accordingly.
As before, \lst|xi| does not need to be processed, since its cached information remains valid.
In contrast to \cref{fig:loops:lam}, $\Var{hi}$ now appears free in \lst|xj| in \cref{fig:loops:lam2}, and is propagated during the fixed-point iteration.
The computed sets are already sound in the second iteration ($\mathtt{run} = 9$), which is confirmed by a third iteration ($\mathtt{run} = 10$).
See \cref{sec:trie} for an efficient representation of these sets and their operations.

\subsection{Nesting \& Well-Formedness}
\label{sec:wf}

\begin{wrapfigure}{r}{.35\textwidth}
\vspace{-.5ex}
\begin{lstlisting}[basicstyle=\ttfamily\footnotesize,xleftmargin=1ex,xrightmargin=1ex,language=Caml]
    f TO LAM int -> int. g 23
    g TO LAM int -> int. h $\Var{f}$
    h TO LAM int -> int. bar $\Var{g}$
\end{lstlisting}
\vspace{-2ex}
\caption{transitive nesting}
\vspace{-2ex}
\label{fig:trans}
\end{wrapfigure}
Ultimately, we must determine the nesting structure induced by a $\lamG$ program.
Consider the program in \cref{fig:trans}:
As $\Var{f}$ appears free in~\lst|g|, \lst|g| must be nested within~\lst|f|.
As $\Var{g}$ appears free in~\lst|h|, \lst|h| must be nested within~\lst|g|.
However, note that $\Var{f}$ does \emph{not} appear free in~\lst|h|.
Nevertheless, \lst|h| must also be nested within~\lst|f|.
Therefore, the nesting relation must be transitive.
We say
$\hl{\ell_1}$ \emph{strictly nests} $\hl{\ell_2}$ if $\hl{\ell_1} \succ_\hl{P} \hl{\ell_2}$, and
$\hl{\ell_1}$ \emph{nests} $\hl{\ell_2}$ if $\hl{\ell_1} \succeq_\hl{P} \hl{\ell_2}$:
\begin{mathpar}
	\irule{N-Strict}{
		\var[\ell_1] \in \FV{\hl{\ell_2}}
	}{
		\hl{\ell_1} \succ_\hl{P} \hl{\ell_2}
	}
	\and
	\irule{N-Trans}{
		\hl{\ell_1} \succ_\hl{P} \hl{\ell_2} \\
		\hl{\ell_2} \succ_\hl{P} \hl{\ell_3}
	}{
		\hl{\ell_1} \succ_\hl{P} \hl{\ell_3}
	}
	\\
	\irule{N}{
		\hl{\ell_1} \succ_\hl{P} \hl{\ell_2}
	}{
		\hl{\ell_1} \succeq_\hl{P} \hl{\ell_2}
	}
	\and
	\irule{N-Refl}{
		\hl{\ell_1} = \hl{\ell_2} \\
	}{
		\hl{\ell_1} \succeq_\hl{P} \hl{\ell_2}
	}
\end{mathpar}

\begin{wrapfigure}{r}{.3\textwidth}
\vspace{-.5ex}
\begin{lstlisting}[basicstyle=\ttfamily\footnotesize,xleftmargin=1ex,xrightmargin=1ex,language=Caml]
    f TO LAM int -> int. $\Var{g}$
    g TO LAM int -> int. $\Var{f}$
\end{lstlisting}
\vspace{-2ex}
\caption{Nonsensical nesting}
\label{fig:nonsense}
\vspace{-1ex}
\end{wrapfigure}
One may construct nonsensical, cyclic nesting dependencies.
For example, \lst|f| must be nested within \lst|g| in \cref{fig:nonsense} as \lst|g|'s variable appears free in \lst|f|.
Conversely, \lst|g| must be nested within \lst|f| as \lst|f|'s variable appears free in \lst|g|:
We prohibit such contradictory programs by requiring the nesting relation to be acyclic, i.e., antisymmetric.
We say the program is \emph{well-formed}:
\begin{mathpar}
	\irule{WF}{
		\forall \hl{\ell_1}, \hl{\ell_2} \in \dom{\hl{P}}:
		\hl{\ell_1} \succeq_\hl{P} \hl{\ell_2} \succeq_\hl{P} \hl{\ell_1} \Rightarrow \hl{\ell_1} = \hl{\ell_2}
	}{
		\wf \hl{P}
	}
\end{mathpar}

\subsubsection{CFG \& Dominance}

A node $n$ \emph{dominates} node $m$ in a directed graph with a unique start node, if all paths from the start to $m$ also go through $n$~\cite{10.1145/1460299.1460314}.
In order to overlay a \ac{CFG} over a $\lamG$ program, we pick a root $\hl{\ell_0}$ and only consider functions as \ac{CFG} nodes $N$ that $\hl{\ell_0}$ nests---much like in a traditional \ac{CFG} where other top-level functions like \lst|printf| do not appear in the \ac{CFG} where they are called:
\begin{equation*}
    \mathsf{CFG}(\hl{P}, \hl{\ell_0}) = \langle N, \leadsto \rangle \quad
    \text{where}\quad
    N = \{\hl{\ell} \in \dom{\hl{P}} \mid \hl{\ell_0} \succeq_\hl{P} \hl{\ell}\} \quad
	\irule{CFG-Succ}{
		\hl{\ell_1} \to_\hl{P} \hl{\ell_2} \\
        \hl{\ell_1},\hl{\ell_2} \in N
	}{
		\hl{\ell_1} \leadsto \hl{\ell_2}
	}
\end{equation*}
If the program is higher-order, the interpretation of the \ac{CFG} becomes unclear.
Nonetheless, we introduce this concept solely to show that nesting implies dominance.
Reconsider \cref{eq:fv_lam}: this equation is the sole location where a free variable is removed.
In other words, a free variable propagates through function expressions unless it refers to its defining function:

\begin{lemma}[One-Step Propagation]\label{lem:prop}
    If
    $\hl{\ell_1} \to_\hl{P} \hl{\ell_2}$ and
    $\hl{\var[\ell]} \in \FV{\hl{\ell_2}}$ and
    $\hl{\ell} \neq \hl{\ell_1}$, then
    $\hl{\var[\ell]} \in \FV{\hl{\ell_1}}$.
\end{lemma}
\begin{proof}
    By definition of $\mathit{FV}$.
\end{proof}
This also holds for paths:
\begin{lemma}[Path Propagation]\label{lem:path}
    If $\hl{\ell_0} \to_\hl{P} \cdots \to_\hl{P} \hl{\ell_k}$ and
    $\hl{\var[\ell]} \in \FV{\hl{\ell_k}}$ and
    $\hl{\ell} \notin \{\hl{\ell_0}, \ldots, \hl{\ell_k}\}$, then
    $\hl{\var[\ell]} \in \FV{\hl{\ell_0}}$.
\end{lemma}
\begin{proof}
    By induction on $k$ and \cref{lem:prop}.
\end{proof}
This insight lets us prove that nesting implies dominance:
\begin{theorem}[Nesting-Dominance]\label{thm:dom}
    If $\wf \hl{P}$ and
    $\hl{\ell_0} \succeq_\hl{P} \hl{\ell_1} \succeq_\hl{P} \hl{\ell_2}$,
	then $\hl{\ell_1}$ dominates $\hl{\ell_2}$ in $\mathsf{CFG}(\hl{P}, \hl{\ell_0})$.
\end{theorem}
\begin{proof}
    Assume $\hl{\ell_0} \ne \hl{\ell_1} \ne \hl{\ell_2}$.
    Otherwise the proof follows from basic properties of dominance. If $\hl{\var[\ell_1]} \in \FV{\hl{\ell_2}}$
    and there exists a path $\hl{\ell_0} \leadsto \cdots \leadsto \hl{\ell_2}$ that does \emph{not} contain~$\hl{\ell_1}$,
    we obtain $\Var{\hl{\ell_1}} \in \FV{\hl{\ell_0}}$ by \cref{lem:path} and, hence, $\hl{\ell_1} \succeq_\hl{P} \hl{\ell_0}$.
    But this contradicts well-formedness, since we also have $\hl{\ell_0} \succeq_\hl{P} \hl{\ell_1}$.
    The case where $\hl{\ell_1} \succ_\hl{P} \hl{\ell_i} \succ_\hl{P} \hl{\ell_2}$ then follows by induction and the transitivity of dominance.
\end{proof}
\begin{remark}
    The reverse does not necessarily hold.
    For example, \lst|bi| dominates \lst|hj| in \cref{fig:loops:dom} but \lst|bi| does not nest \lst|hj| in \cref{fig:loops:nest}.
    Additionally, note that $\hl{\ell_0}$ may have free variables.
\end{remark}
Intuitively, well-formedness ensures that every variable use is properly \enquote{enclosed} by its definition---analogous to \cref{prop:dom}.
However, our formulation does not rely on an explicit \ac{CFG} or dominance relation.
Instead, it uses the nesting relation, which naturally extends to higher-order programs where control flow is implicit.
Accordingly, we introduce a structural property that paraphrases \irefrule{WF} and generalizes dominance beyond first-order control flow:
\begin{property}[Nesting]\label{prop:nesting}
    Nesting must be antisymmetric.
\end{property}

\subsection{Substitution}
\label{sec:subst}

Substitution replaces free occurrences of a variable by an expression.
In \lamG, substitution is more intricate than in the standard $\lambda$-calculus,
since \lamG lacks explicit block nesting.
Conversely, the structural nesting that does exist in a program may overapproximate the actual data dependencies between expressions---much like the dominator tree in a \ac{CFG}.
The substitution procedure defined in this section rewrites exactly those expressions that (transitively) depend on the substituted variable.

The central question in substitution is therefore whether a given expression~$\hl{e}$ must be recursively rewritten or can be left unchanged.
This decision depends on the free-variable set of~$\hl{e}$:
if it contains the variable being substituted, $\hl{e}$ must be rewritten.
As we traverse subexpressions, we may encounter functions whose bodies reference the substitution variable;
in such cases, new functions with rewritten bodies are created.
Any expression in which variables of these new functions occur free must also be updated accordingly.
This mirrors the transitivity of the nesting relation (\irefrule{N-Trans}).

To formalize this dependency-sensitive behavior, substitution must keep track of both variable and function rewrites.
Specifically, we maintain a \emph{variable map}~$\hl{V}$ that maps variables to their replacement expressions, and a \emph{function map}~$\hl{F}$ that maps functions to their rewritten counterparts.
Since substitution may introduce new label-to-function mappings, the program~$\hl{P}$ must also be threaded through the process.
Hence, substitution is defined as a function
\begin{align}
	\llangle \hl{V}, \hl{F}, \hl{P}, \hl{e} \rrangle          & = \langle \hl{F'}, \hl{P'}, \hl{e'} \rangle
    \intertext{%
        which recursively rewrites an expression while threading the function and program maps.
        If the domain of~$\hl{V}$ (the set of substitution variables) and the free variables of~$\hl{e}$ are disjoint, no rewriting is required:%
    }
  \label{subst:intersect}
	\llangle \hl{V}, \hl{F}, \hl{P}, \hl{e} \rrangle          & = \langle \hl{F}, \hl{P}, \hl{e} \rangle                                                         &  & \text{if $\FV{\hl{e}} \cap \dom{\hl{V}} = \varnothing$}                                                                                                                                                                        \\
	\intertext{Otherwise, substitution proceeds recursively over all subexpressions:}
	\llangle \hl{V}, \hl{F}, \hl{P}, \hl{e_1\ e_2} \rrangle & = \langle \hl{F''}, \hl{P''}, \hl{e'_1}\ \hl{e'_2} \rangle                                                 &  & \text{where $\llangle \hl{V}, \hl{F\phantom{'}}, \hl{P\phantom{'}}, \hl{e_1} \rrangle = \langle \hl{F'\phantom{'}}, \hl{P'\phantom{'}}, \hl{e'_1} \rangle$}                                                                    \\
	                                                        &                                                                                                  &  & \text{\phantomr{where}{and} $\llangle \hl{V}, \hl{F'}, \hl{P'}, \hl{e_2} \rrangle = \langle \hl{F''}, \hl{P''}, \hl{e'_2} \rangle$}\notag                                                                                      \\
    \intertext{%
        If a free variable is encountered, its substitute is obtained directly from~$\hl{V}$.
        By \cref{subst:intersect}, it must be present in~$\hl{V}$:%
    }
	\llangle \hl{V}, \hl{F}, \hl{P}, \hl{\var} \rrangle     & = \langle \hl{F}, \hl{P}, \hl{V}(\var) \rangle                                                                                                                                                                                                                                                                                       \\
    \intertext{If a function expression is encountered, we either reuse its existing substitution}
	\label{subst:fun1}
	\llangle \hl{V}, \hl{F}, \hl{P}, \hl{\ell} \rrangle     & = \langle \hl{F}, \hl{P}, \hl{F}(\hl{\ell}) \rangle                                                   &  & \text{if $\hl{\ell} \in \dom{\hl{F}}$}                                                                                                                                                                                         \\
    \intertext{or create a new function with a rewritten body:}
	\label{subst:fun2}
	\llangle \hl{V}, \hl{F}, \hl{P}, \hl{\ell} \rrangle     & = \langle \hl{F'}, \hl{P'}[\hl{\ell'} \mapsto \hl{\lam{t}[u]{e'}}], \hl{\ell'} \rangle &  & \text{where $\hl{P}(\hl{\ell}) = \hl{\lam{t}[u]{e}}$, $\hl{\ell'}$ fresh,}                                                                                                                                                  \\
	                                                        &                                                                                                  &  & \mathclap{\text{and $\llangle \hl{V}[\var \mapsto \var[\ell']], \hl{F}[\hl{\ell} \mapsto \hl{\ell'}], \hl{P}[\hl{\ell'} \mapsto \hl{\lam{t}[u]{\unset}}], \hl{e} \rrangle = \langle \hl{F'}, \hl{P'}, \hl{e'} \rangle$}}\notag
\end{align}
In \cref{subst:fun2}, we look up the function bound to~$\hl{\ell}$ in~$\hl{P}$,
create a fresh label~$\hl{\ell'}$ mapped to a stub function with an unset body,
and extend the variable and function maps so that occurrences of the old variable and function
are redirected to their new counterparts.
The unset body serves as a placeholder that resolves recursion:
further rewrites will replace all occurrences of~$\hl{\ell}$ with~$\hl{\ell'}$ (see \cref{subst:fun1}).
Once the rewritten body~$\hl{e'}$ is obtained, it replaces the unset body of~$\hl{\ell'}$.

Substitution preserves well-formedness and typing:
\begin{lemma}[Substitution]\label{lem:subst}
    Let $\llangle \{\var \mapsto \hl{e_\ell}\}, \varnothing, \hl{P}, \hl{e} \rrangle = \langle \hl{F}, \hl{P'}, \hl{e'} \rangle$.
    If
    $\wf \hl{P}$ and
    $\vdash \hl{P}$ and
    $\hl{P} \vdash \hl{e} : \hl{t}$ and
    $\hl{P} \vdash \hl{e_\ell} : \hl{t_\ell}$, then
    $\wf    \hl{P'}$ and
    $\vdash \hl{P'}$ and
    $\hl{P'} \vdash \hl{e'} : \hl{t}$.
\end{lemma}
\begin{proof}
    By induction on a derivation of $\wf \hl{P}$, $\vdash \hl{P}$, and $\hl{P} \vdash \hl{e} : \hl{t}$.
\end{proof}

\subsection[beta-Reduction]{$\beta$-Reduction}
\label{sec:beta}

Building on substitution, we now define $\beta$-reduction
\begin{mathpar}
	\beta\inferrule{
        \hl{P}(\hl{\ell}) = \hl{\lam{t}[u]{e}} \\
        \llangle \{\var \mapsto \hl{e_v}\}, \varnothing, \hl{P}, \hl{e} \rrangle = \langle \hl{F'}, \hl{P'}, \hl{e'_b} \rangle
	}{
		\hl{P}, \hl{\ell\ e_v} \to_\beta  \hl{P'},\hl{e'_b}
	}
\end{mathpar}
and define full evaluation using a standard left-to-right, strict evaluation order:
\begin{mathpar}
    \irule{V-Fun}{
    }{
        \val{\hl{\ell}}
    }
    \and
    \irule{V-Const}{
    }{
        \val{\hl{c}}
    }
    \and
	\irule{E-App1}{
		\hl{P}, \hl{e_1} \to  \hl{P'},\hl{e'_1}
	}{
		\hl{P}, \hl{e_1\ e_2} \to  \hl{P'},\hl{e'_1\ e_2}
	}
    \and
	\irule{E-App2}{
        \val{\hl{e_1}} \\
		\hl{P}, \hl{e_2} \to  \hl{P'},\hl{e'_2}
	}{
		\hl{P}, \hl{e_1\ e_2} \to  \hl{P'},\hl{e_1\ e'_2}
	}
    \and
	\textsc{E-}\beta\inferrule{
        \val{\hl{e_v}} \\
		\hl{P}, \hl{\ell\ e_v} \to_\beta  \hl{P'},\hl{e'_b}
	}{
		\hl{P}, \hl{\ell\ e_v} \to  \hl{P'},\hl{e'_b}
	}
\end{mathpar}

\begin{theorem}[Progress]
    If
    $\hl{P} \vdash \hl{e} : \hl{t}$ and
    $\FV{\hl{e}} = \varnothing$, then either
    $\val{\hl{e}}$ or $\exists \hl{P'},\hl{e'}: \hl{P},\hl{e} \rightarrow \hl{P'},\hl{e'}$.
\end{theorem}
\begin{remark}
    Since $\hl{e}$ is closed, it follows from \cref{lem:path} that the part of the program $\hl{P}$ reachable from $\hl{e}$ is well-formed.
\end{remark}
\begin{proof}
	By induction on a derivation of $\hl{P} \vdash \hl{e} : \hl{t}$.
\end{proof}
\begin{theorem}[Preservation]
    If
    $\wf \hl{P}$ and
    $\vdash \hl{P}$ and
    $\hl{P} \vdash \hl{e} : \hl{t}$ and
    $\hl{P},\hl{e} \rightarrow \hl{P'},\hl{e'}$, then
    $\wf \hl{P'}$ and
    $\vdash \hl{P'}$ and
    $\hl{P'} \vdash \hl{e'} : \hl{t}$.
\end{theorem}
\begin{proof}
	By induction on a derivation of $\hl{P},\hl{e} \rightarrow \hl{P'},\hl{e'}$ and \cref{lem:subst}.
\end{proof}

\subsection[eta-Conversion]{$\eta$-Conversion}
\label{sec:eta}

Using the notion of free variables, we define $\eta$-reduction.
Dually, for any expression of function type, we can perform $\eta$-expansion, wrapping it into an equivalent function expression:
\begin{mathpar}
	\eta\textsc{-Red}\inferrule{
        \hl{P}(\hl{\ell}) = \hl{\lam{t}[u]{e\ \var}} \\\\
        \var \notin \FV{\hl{e}}
	}{
		\hl{P}, \hl{\ell} \to_\eta \hl{P},\hl{e}
	}
    \and
	\eta\textsc{-Exp}\inferrule{
        \hl{P'} = \hl{P}[\hl{\ell} \mapsto \hl{\lam{t}[u]{e\ \var}}] \\\\
        \text{$\hl{\ell}$ fresh} \\
        \hl{P} \vdash \hl{e} : \hl{t \to u}
	}{
        \hl{P},\hl{e} \leftarrow_\eta \hl{P'}, \hl{\ell}
	}
\end{mathpar}

\subsubsection[Critical Edge Elimination]{Critical Edge Elimination}
\label{sec:crit}

\input{fig/eta}

A \emph{critical edge} in a \ac{CFG} is an edge whose source block has multiple successors and whose destination block has multiple predecessors.
Critical edges hinder many analyses and optimizations and are typically eliminated by inserting an empty intermediate block.
For instance, removing \lst|bi| in \cref{fig:loops:ssa} would introduce a critical edge.


We now examine the corresponding situation in \lamG and remove \lst|bi| in \cref{fig:loops:lam}.
Eliminating \lst|bi| yields the expression \lst|br$_\hl{\bot}$(i1<n, hj, xi)|, which is ill-typed since \lst|hj| has type \lst|int -> BOT|: the argument \lst|i1| to \lst|hj| is lost.
Thus, when translating an SSA-form program to \lamG, all critical edges must be eliminated beforehand.
Alternatively, one may use a branching function with a polymorphic type \lst|FORALLT.FORALLU.[bool, T -> BOT, U -> BOT, T, U] -> BOT| to pass $\phi$-function arguments explicitly: \lst|br(i1<n, hj, xi, i1, ())|.

The key observation is that function \lst|hj| is used \emph{both} as a callee and as an argument to another function (\lst|br|).
We call a function expression \emph{known} if it appears as the callee of an application, and \emph{unknown} otherwise.
A function is \emph{well-known} if all of its function expressions are known~\cite{10.5555/889478}.
Accordingly, \lst|hj| is no longer well-known as it was in the original program before the critical edge was split via \lst|bi|.

More generally, optimizations may wish to change a function’s type---for example, to remove dead parameters.
Consider \cref{fig:eta1}, where the second parameter of \lst|f| is dead.
We could create a specialized version \lst|f'| that omits this parameter, but doing so would render \lst|g f'| ill-typed.
However, if we first $\eta$-expand \lst|f| to \lst|f$_\eta$| and rewrite \lst|g f| as \lst|g f$_\eta$|, then \lst|f| becomes well-known (\cref{fig:eta2}).
At that point, it is safe to remove the unused parameter (yielding \lst|f'|, \cref{fig:eta3}) and to adjust all known function expressions, \lst|f (x, y)| to \lst|f' x| in this example.

In summary, if a function has both known and unknown function expressions, we can construct an $\eta$-expanded version and substitute it for all previously unknown occurrences.
This transformation makes the original function well-known.
For first-order programs that use only a branching function, this achieves precisely the effect of critical-edge elimination; for higher-order programs, it generalizes that concept.

\subsection{Other Binders and Dependent Types}
\label{sec:dep_types}

This framework naturally extends to dependent types, where types are themselves expressions:
When gathering information, we must also traverse the corresponding types.
For example, \cref{eq:lv_app} is adjusted as follows:
\begin{equation}
	\LV{\hl{e_1\ e_2}} = \LV{\hl{e_1}} \cup \LV{\hl{e_2}} \cup \LV{\hl{e_t}} \qquad \text{where $\hl{P} \vdash \hl{e_1}\ \hl{e_2} : \hl{e_t}$}
\end{equation}

It is straightforward to support additional binders such as dependent pairs.
The computations introduced above just have to be generalized to operate on arbitrary binders rather than functions alone.
Our implementation in \mimir already supports dependent types and several other binders.
The other algorithms presented in this paper are extended accordingly to work on dependently typed expressions and/or other binders.

%% file: fig/eta.tex
\begin{figure}[t]
	\begin{subcaptionblock}{.315\textwidth}
		\begin{lstlisting}[basicstyle=\ttfamily\tiny,xleftmargin=1ex,xrightmargin=1ex,language=lamg]
            f  TO LAM [int, int] -> int. $\Var{f}$.0

            g  TO LAM([int, int] -> int) -> int. B
            h  TO LAM [] -> int.
                 f (x, y) + g f
        \end{lstlisting}
        \vspace{-1ex}
		\caption{\texttt{f} used as callee and argument}
		\label{fig:eta1}
	\end{subcaptionblock}
	\begin{subcaptionblock}{.315\textwidth}
		\begin{lstlisting}[basicstyle=\ttfamily\tiny,xleftmargin=1ex,xrightmargin=1ex,language=lamg]
            f  TO LAM [int, int] -> int. $\Var{f}$.0
            f$_\eta$ TO LAM [int, int] -> int. f $\Var{f_\eta}$
            g  TO LAM([int, int] -> int) -> int. B
            h  TO LAM [] -> int.
                 f (x, y) + g f$_\eta$
        \end{lstlisting}
        \vspace{-1ex}
		\caption{\texttt{f} $\eta$-expanded to \texttt{f}$_\eta$}
		\label{fig:eta2}
	\end{subcaptionblock}
	\begin{subcaptionblock}{.345\textwidth}
		\begin{lstlisting}[basicstyle=\ttfamily\tiny,xleftmargin=1ex,xrightmargin=1ex,language=lamg]
            f' TO LAM int -> int. $\Var{f'}$
            f$_\eta$ TO LAM [int, int] -> int. f' ($\Var{f_\eta}$.0)
            g  TO LAM([int, int] -> int) -> int. B
            h  TO LAM [] -> int.
                 f' x + g f$_\eta$
        \end{lstlisting}
        \vspace{-1ex}
		\caption{\texttt{f} now optimized to \texttt{f'}}
		\label{fig:eta3}
	\end{subcaptionblock}
    \vspace{-2ex}
    \caption{$\eta$-Expansion makes \emph{known} functions \emph{well-known}; this allows us to adjust calls to the original function.}
\end{figure}

%% file: nest.tex
\section{The Nesting Tree}
\label{sec:nest_tree}

While information about free variables is sufficient for many tasks within an \ac{IR}, the \emph{nesting tree} provides additional structural information.
This data structure captures the minimal nesting required to reconstruct lexical scoping.

If the program is well-formed, we can represent the nesting relation as a forest of trees, where each tree is a relaxed dominator tree (\cref{thm:dom}).
This is similar to \acp{CFG}, where each function has its own dominator tree.
Analogous to the immediate dominator in a \ac{CFG}, we introduce the concept of an \emph{immediate nester}:
\begin{equation}
	\inest{\hl{\ell}} =
	\begin{cases}
		\hl{\ell_\mathit{imm}} & \text{if $\hl{\ell_\mathit{imm}} \succ_\hl{P} \hl{\ell}$ and $\nexists \hl{\ell'}: \hl{\ell_\mathit{imm}} \succ_\hl{P} \hl{\ell'} \succ_\hl{P} \hl{\ell}$} \\
		\bot                   & \text{if no such $\hl{\ell_\mathit{imm}}$ exists ($\hl{\ell}$ is root)}
	\end{cases}
\end{equation}

\input{fig/nest}

To construct the nesting tree (\cref{fig:nest}), we traverse all functions reachable from a chosen root $\hl{\ell_\mathtt{root}}$, following the nesting relation.
For each newly discovered function $\hl{\ell}$, we attempt to place it as deep as possible:
starting from the current node \lst|n|, we check whether this node's variable is free in $\hl{\ell}$.
If not, we move up the partially constructed tree and repeat until we find the deepest ancestor whose variable is free in $\hl{\ell}$.
We then attach $\hl{\ell}$ as a child of that ancestor.
Such an ancestor always exists because $\Var{\ell_\mathtt{root}}$ is free in all reachable functions.

\subsection[Translation from Lambda-G to a Lexically Scoped Language]{Case Study: Translation from \lamG to a Lexically Scoped Language}

The nesting tree enables more structured analyses, which are best understood through a concrete example.
We therefore describe how to translate a \lamG program into a functional language with explicit lexical scoping, such as \ocaml.
We consider two running examples:
program~A (\cref{fig:loops:lam}/\subref{fig:loops:ml}/\subref{fig:loops:nest}) and
program~B (\cref{fig:loops:lam2}/\subref{fig:loops:ml2}/\subref{fig:loops:nest2}).
Intuitively, the nesting tree determines how the \enquote{soup of functions} in $\lamG$ must be arranged into lexical scopes so that free variables can access their definitions.
For example, in programs A and B, function \lst|hi| refers to \lst|n = $\Var{f}$.0|.
Consequently, \lst|hi| must be nested within \lst|f| in the \ocaml translation.

\subsubsection{Sibling Dependencies}
\label{sec:siblings}

We must also determine the order in which functions at the same nesting level are defined.
For example, since \lst|bi| invokes \lst|hj| in program B, \lst|hj| must be defined before \lst|bi|.
To formalize this, we introduce \emph{sibling dependencies}, which capture how functions at the same level depend on one another.
These dependencies induce an ordering on functions that share the same immediate nester.
Formally, there is a sibling dependency from $\hl{\ell_1}$ to $\hl{\ell_2}$ if some function $\hl{\ell_1'}$ nested within $\hl{\ell_1}$ uses $\hl{\ell_2}$, and $\hl{\ell_1}$ and $\hl{\ell_2}$ have the same immediate nester:
\begin{mathpar}
	\irule{Sibl}{
		\hl{\ell_1} \succeq_\hl{P} \hl{\ell'_1} \\
		\hl{\ell'_1} \to_\hl{P} \hl{\ell_2} \\
		\inest{\hl{\ell_1}} = \inest{\hl{\ell_2}}
	}{
		\hl{\ell_1} \dep_\hl{P} \hl{\ell_2}
	}
\end{mathpar}
The continuation \lst|xi| in programs A and B, for instance, does not depend on any of its siblings.
It may therefore be placed at any position among them, provided it remains nested within \lst|hi|.

\subsubsection{Direct \& Mutual Recursion}

When sibling dependencies form cycles, they induce recursion.
Acyclic sibling dependencies indicate that one subtree depends on another, but not vice versa.
A self-loop denotes direct recursion, while longer cycles correspond to mutual recursion.

To identify these cycles, we run Tarjan’s \ac{SCC} algorithm~\cite{DBLP:journals/siamcomp/Tarjan72} \emph{independently} on the sibling dependencies \emph{at each level} of the nesting tree.
This identifies \lst|hj| as an inner loop of \lst|hi| in program~B.
Only in the degenerate case of program~A does the nesting tree merge both loops into a mutually recursive set of functions.
As a useful byproduct, \citeauthor{DBLP:journals/siamcomp/Tarjan72}'s algorithm also produces a topological ordering of the resulting \acp{SCC}.
This enriched structure allows us to traverse the nesting tree while tracking loop levels.

In the \ocaml translation, these cycles give rise to recursive bindings:
a self-cycle produces a \lst[language=Caml]|let rec| definition, whereas longer cycles require \lst[language=Caml]|and|.

\paragraph{Irreducible Control Flow}

\emph{Natural loops} are single-entry \acp{SCC} in a \ac{CFG}: they have exactly one header node that \emph{dominates} every node in the loop.
\emph{Irreducible loops} have multiple entry points (hence no single dominating header) and are therefore not natural loops.
This reliance on dominance is a major reason why irreducible loops require special handling in classical \ac{CFG}-based compilers.
In the nesting tree, irreducible loops manifest as larger \acp{SCC} among sibling dependencies.
Our analysis identifies them without special treatment, although subsequent loop optimizations may still require additional care.

\paragraph{Virtual Root}

We also introduce the concept of a \emph{virtual root} that, by definition, nests every function.
This makes it possible to construct a single nesting tree for the entire program, in which all closed functions become direct children of the virtual root.
The sibling dependencies among these functions then induce a traditional call graph; our standard \ac{SCC} detection directly identifies (mutual) recursion.

\subsection{Case Study: Code Motion}
\label{sec:code-place}

As another use case, we consider code motion~\cite{DBLP:journals/toplas/ClickC95}.
In a traditional \ac{SSA}-based \ac{CFG}, \emph{early placement} places each expression as early as possible by first placing its subexpressions;
the deepest of these placements in the dominator tree determines the earliest valid position of the expression itself.
Conversely, \emph{late placement} places expressions as late as possible by first placing all of their users;
the least common ancestor of these placements in the dominator tree yields the latest valid placement.

Both placement strategies can be implemented directly using the nesting tree:
for early placement, we consult an expression's level in the nesting tree;
for late placement, we compute the least common ancestor of its users within the nesting tree.

In addition, the program may contain functions without a variable---such as \lst|bi| in programs A and B---which are leaves of the nesting tree, often introduced during critical-edge elimination (\cref{sec:crit}).
If we want to place expressions there as well, we must also consider sibling dependencies.

%% file: fig/nest.tex
\begin{figure}[t]
	\begin{lstlisting}[language=C,morekeywords={not,or,and,True,False},basicstyle=\tiny\ttfamily]
        queue.push(make_node($\ell_\mathtt{root}$, $\bot$))
        vars $\shortleftarrow$ $\{\Var{\ell_\mathtt{root}}\}$
        while not queue.empty():
            n $\shortleftarrow$ queue.pop()
            for $\ell$ $\in$ LF(n.function.body):
                if not already_constructed($\ell$):
                    for (parent $\shortleftarrow$ n; parent $\neq \bot$; parent $\shortleftarrow$ inest(parent)):
                        if $\Var{\ell_\mathtt{parent}}$ $\in$ FV($\ell$):
                            vars $\shortleftarrow$ vars $\cup$ $\{\Var{\ell}\}$
                            queue.push(make_node($\ell$, parent))
                            break
    \end{lstlisting}
    \vspace{-2ex}
    \caption{Nesting tree construction}
    \label{fig:nest}
\end{figure}

%% file: trie.tex
\section{Efficient Set Operations}
\label{sec:trie}

As discussed in \cref{sec:fvs}, \emph{each expression} stores the sets of local variables ($\mathit{LV}$) and local functions ($\mathit{LF}$).
Furthermore, \emph{each function} stores the set of functions that reference it ($\mathit{UF}$) as well as a cache of its free variables ($\mathit{FV}$).
If this information were stored separately for each expression/function, the space requirements could easily explode.
To avoid this blow‑up, this section introduces a shared data structure that expressions reference rather than duplicate.
Each expression/function therefore stores only a single pointer to its corresponding sets.
As a result, copying reduces to a pointer copy, and checking for (in‑)equality becomes a pointer comparison, which is required to determine whether the iteration described in \cref{sec:fv_impl} has reached a fixed point.

We first describe a naive implementation using hash-consed, ordered arrays, which performs well for small sets.
Subsequently, we introduce an ordered trie with better space requirements.
Finally, we index paths in this trie with a splay tree, as in a link-cut tree, to improve asymptotic runtime behavior.
We consider the operations of insertion, union, removal, membership, and intersection testing.
These are the fundamental set operations required to implement the algorithms discussed in this paper.

\subsection{Hash-consed, Ordered Arrays}

This approach stores each set as an immutable, sorted array of unique elements.
Each array is interned in a hash set, a technique known as \emph{hash-consing}.
Hash-consing ensures that identical arrays are stored only once globally, so equality and sharing reduce to constant-time pointer comparison.

\emph{Inserting} an element into a set creates a new ordered array with the element placed in sorted order.
\emph{Removing} an element constructs a new array without that element.
\emph{Union} is computed via a linear-time merge of the two arrays, since the arrays are ordered.
In each case, the resulting array is hash-consed, and the overall runtime scales linearly with the size of the set, plus the cost of hashing.
\emph{Membership} can be tested via binary search over the ordered array, yielding logarithmic-time complexity.
However, as we use this data structure only for small sets, we found that a linear scan is slightly faster in practice.
To test whether two sets intersect, we traverse the ordered arrays simultaneously, stopping as soon as a shared element is found.
Other algorithms offer better average-case performance but come with higher constant overhead.
Again, as we use this approach only for small sets, we opt for the simpler implementation.

This representation is simple and compact, with very low constant factors.
We arena-allocate the arrays, which makes memory allocation cheap and array accesses more cache-friendly.

The abstract domain of the free variable analysis is the power set $\mathcal{P}(V)$ of the set $V$ of $v \coloneqq |V|$ program variables.
If all elements of this domain were materialized, the total space required by the array representation would be
\begin{equation}
    \sum_{i=0}^{v} i \binom{v}{i} = v 2^{v-1} = \mathcal{O}(v 2^{v})\ .
\end{equation}
In practice, however, the analysis constructs only those elements of $\mathcal{P}(V)$ that arise from analyzing program subexpressions.
Since a program of size $n$ has $\mathcal{O}(n)$ subexpressions, and each subexpression gives rise to at most one new domain element, at most $\mathcal{O}(n)$ distinct domain elements are materialized.
As each such element contains at most $n$ variables, the resulting space bound is $\mathcal{O}(n^2)$.
Likewise, we need at most $\mathcal{O}(n^2)$ space to keep track of functions.

\subsection{Ordered Trie}

\input{fig/trie}
A more space-efficient alternative is a globally shared ordered trie that represents elements of the same abstract domain.
Each element of $\mathcal{P}(V)$ corresponds to a unique node in a full trie, where the sequence of elements in the set is encoded by the path from that node to the root.

For example, the set $\{a,b,c,d\}$ is identified by the node labeled $d_8$ of the ordered trie in \cref{fig:orderedtrie}.
Walking up to the root~$\ocircle$ yields $c_4$, $b_2$, and $a_1$.
The set $\{a,b,c\}$ is identified by the node labeled $c_4$, reusing the same prefix nodes $b_2$ and $a_1$.

Unlike the array-based representation, the trie shares common prefixes between sets.
Consequently, if all elements of the abstract domain $\mathcal{P}(V)$ were materialized, the full trie would have exactly $2^{v}$ nodes.
Under the same assumptions as above, however, only $\mathcal{O}(n)$ domain elements arise during the analysis of a program of size $n$.
In this case, the trie requires only $\mathcal{O}(n)$ space, improving upon the quadratic bound of the array-based representation under the same assumptions.
Likewise, we need at most $\mathcal{O}(n)$ nodes in a trie to keep track of functions.

However, although the sets are ordered, membership tests are linear in the depth of the trie because we do not have random access to the nodes along the path.

\subsection{Indexed Trie}

A \emph{splay tree}~\cite{DBLP:journals/jacm/SleatorT85} is a variant of a binary search tree that is self-adjusting.
Instead of maintaining a fixed structure, it rotates an accessed node to the root through a process called \emph{splaying}.
This operation moves frequently accessed nodes closer to the root, improving access time for those elements over time.

A \emph{link-cut tree}~\cite{DBLP:journals/jcss/SleatorT83} is a data structure that supports dynamic trees, allowing efficient access to and updates of tree paths.
It uses a series of splay trees to represent a forest of rooted trees, enabling operations such as path queries and tree merges to be performed in \emph{amortized} $\mathcal{O}(\log n)$ time.

Returning to the ordered trie, we \emph{index} its paths using \emph{auxiliary splay trees}, similar to a link-cut tree (see \crefrange{fig:indexedtrie}{fig:auxtrees}).
We never require the cut operation and only link nodes incrementally.

The trie is decomposed into \emph{preferred paths} by assigning each node at most one \emph{preferred child}, namely the child in the most recently accessed subtree.
Preferred paths change dynamically with access patterns.
Each preferred path is represented by an auxiliary splay tree whose nodes are keyed by depth in the trie.
These auxiliary trees allow efficient navigation along paths by skipping large portions of them.
To reconnect the auxiliary trees into the original trie structure, we use \emph{path-parent} pointers.
For each preferred path, exactly one node carries such a pointer: the node that is currently the root of the corresponding auxiliary splay tree.
This node stores a pointer to its parent in the trie, which lies on a different preferred path.

This design reduces \emph{membership test} complexity from linear to amortized logarithmic time.
It also provides an opportunity to short-circuit insertions (removals), achieving logarithmic time whenever the element is already present (absent).

\begin{figure}[t]
\begin{lstlisting}[language=Python,morekeywords={True,False},basicstyle=\tiny\ttfamily]
while not n1.is_root() and not n2.is_root():
    if n1.key < n2.key:
        # Search for n2.key in n1's splay structure
        n1 $\shortleftarrow$ n1.find(n2.key) # Find n2.key or the element just greater than that
        if n1.key $=$ n2.key:
            return True
        n1 $\shortleftarrow$ n1.parent # Move upward
    else:
        # Symmetric case: search for n1.key in n2's splay structure
return False
\end{lstlisting}
\vspace{-2ex}
\caption{Intersection test for two sets organized in an indexed trie}
\label{fig:intersect}
\end{figure}

\emph{Intersection tests} are required to implement \cref{subst:intersect}.
Although they remain linear in the worst case, the use of a link-cut tree enables several optimizations.
First, we can detect shared prefixes in logarithmic time via a least common ancestor query.
If one of the sets is small (16 elements or fewer---see \cref{sec:impl}), we perform at most 16 membership tests.
Otherwise, we execute an alternating splay search:
at each step, the current element of one set is located in the other via its splay structure, after which the roles are swapped.
This symmetric strategy quickly discards large disjoint regions (see \cref{fig:intersect}).
This method achieves a runtime of $\mathcal{O}(\log n + \log m)$ when the number of skipped elements dominates.

\subsubsection{Implementation}
\label{sec:impl}

Our implementation distinguishes four cases:
empty sets are represented by a null pointer;
singleton sets use a pointer that directly references the single element;
small sets with up to 16 elements are stored as hash-consed, ordered arrays;
and larger sets are represented using an indexed trie.
Empirically, we found that this threshold (16 elements) marks the point at which the indexed trie outperforms the naive array-based implementation.
We again use arena allocation for all trie nodes.

\mimir maintains two global, immutable data structures: one for variables and one for functions.
\emph{All variable and function sets reference these data structures.}
For example, all occurrences of the set $\{\Var{f}, \Var{hi}\}$ in \cref{fig:fvs} refer to the same memory location.

The ordering of trie nodes is arbitrary;
therefore, we assign each node an increasing counter as its sort key.
This design enables an additional optimization: nodes without a sort key have not yet been inserted into the trie and can thus be discarded quickly during operations such as membership tests.
Moreover, each trie node records the minimum key along the path to the root, enabling further short‑circuiting operations.
Incrementing the counter only when necessary also helps keep the trie structure flat and efficient.

%% file: fig/trie.tex
\begin{figure}[t]
	\begin{subcaptionblock}{.25\textwidth}
		\begin{tikzpicture}
			\node[m] (ro) { $\ocircle$ };

			\node[m] (b)  [below=of ro] {$b_1$};
			\node[m] (a)  [ left=of b]  {$a_1$};
			\node[m] (c)  [right=of b]  {$c_1$};
			\node[m] (d)  [right=of c]  {$d_1$};

			\node[m] (ax) [below=of a ] { };
			\node[m] (ad) [ left=of ax] {$d_2$};
			\node[m] (ac) [ left=of ad] {$c_2$};
			\node[m] (ab) [ left=of ac] {$b_2$};

			\node[m] (bd) [below=of b ] {$d_3$};
			\node[m] (bc) [ left=of bd] {$c_3$};

			\node[m] (cd) [below=of c] {$d_4$};

			\node[m] (abd) [below=of ab ] {$d_5$};
			\node[m] (abc) [ left=of abd] {$c_4$};

			\node[m] (acd) [below=of ac] {$d_6$};

			\node[m] (bcd) [below=of bc] {$d_7$};

			\node[m] (abcd) [below=of abc] {$d_8$};

			\draw (ro) -- (a);
			\draw (ro) -- (b);
			\draw (ro) -- (c);
			\draw (ro) -- (d);

			\draw (a) -- (ab);
			\draw (a) -- (ac);
			\draw (a) -- (ad);

			\draw (b) -- (bc);
			\draw (b) -- (bd);

			\draw (c) -- (cd);

			\draw (ab) -- (abc);
			\draw (ab) -- (abd);

			\draw (ac) -- (acd);

			\draw (bc) -- (bcd);

			\draw (abc) -- (abcd);
		\end{tikzpicture}
        \vspace{-3ex}
		\caption{Ordered trie}
		\label{fig:orderedtrie}
	\end{subcaptionblock}
	\begin{subcaptionblock}{.25\textwidth}
		\begin{tikzpicture}
			\node[m] (ro) { $\ocircle$ };

			\node[m] (b)  [below=of ro] {$b_1$};
			\node[m] (a)  [ left=of b]  {$a_1$};
			\node[m] (c)  [right=of b]  {$c_1$};
			\node[m] (d)  [right=of c]  {$d_1$};

			\node[m] (ax) [below=of a ] { };
			\node[m] (ad) [ left=of ax] {$d_2$};
			\node[m] (ac) [ left=of ad] {$c_2$};
			\node[m] (ab) [ left=of ac] {$b_2$};

			\node[m] (bd) [below=of b ] {$d_3$};
			\node[m] (bc) [ left=of bd] {$c_3$};

			\node[m] (cd) [below=of c] {$d_4$};

			\node[m] (abd) [below=of ab ] {$d_5$};
			\node[m] (abc) [ left=of abd] {$c_4$};

			\node[m] (acd) [below=of ac] {$d_6$};

			\node[m] (bcd) [below=of bc] {$d_7$};

			\node[m] (abcd) [below=of abc] {$d_8$};

			\draw[very thick] (ro) -- (a);
			\draw (ro) -- (b);
			\draw (ro) -- (c);
			\draw (ro) -- (d);

			\draw[very thick] (a) -- (ab);
			\draw (a) -- (ac);
			\draw (a) -- (ad);

			\draw[very thick] (b) -- (bc);
			\draw (b) -- (bd);

			\draw[very thick] (c) -- (cd);

			\draw[very thick] (ab) -- (abc);
			\draw (ab) -- (abd);

			\draw[very thick] (ac) -- (acd);

			\draw[very thick] (bc) -- (bcd);

			\draw[very thick] (abc) -- (abcd);
		\end{tikzpicture}
        \vspace{-3ex}
		\caption{Indexed trie}
		\label{fig:indexedtrie}
	\end{subcaptionblock}
	\begin{subcaptionblock}{.25\textwidth}
		\begin{tikzpicture}
			\node[m] (a1) { $a_1$ };

			\node[m] (xx)  [below     =of a1] {};
			\node[m] (d6)  [ left     =of xx] {$d_6$};
			\node[m] (c4)  [ left     =of d6] {$c_4$};
			\node[m] (d2)  [right     =of xx] {$d_2$};
			\node[m] (ro)  [right     =of d2] {$\ocircle$};

			\node[m] (d8) [below  left=of c4] {$d_8$};
			\node[m] (b2) [below right=of c4] {$b_2$};
			\node[m] (c2) [below right=of d6] {$c_2$};
			\node[m] (d1) [below  left=of ro] {$d_1$};
			\node[m] (c3) [below      =of ro] {$c_3$};
			\node[m] (d4) [below right=of ro] {$d_4$};

			\node[m] (d5) [below      =of b2] {$d_5$};
			\node[m] (d7) [below  left=of c3] {$d_7$};
			\node[m] (b1) [below right=of c3] {$b_1$};
			\node[m] (c1) [below right=of d4] {$c_1$};

			\node[m] (d3) [below      =of b1] {$d_3$};

			\draw[very thick] (c4) -- (a1);
			\draw[very thick] (ro) -- (a1);
			\draw[a]          (d6) -- (a1);
			\draw[a]          (d2) -- (a1);

			\draw[very thick] (d8) -- (c4);
			\draw[very thick] (b2) -- (c4);
			\draw[very thick] (c2) -- (d6);
			\draw[a         ] (d1) -- (ro);
			\draw[a         ] (c3) -- (ro);
			\draw[a         ] (d4) -- (ro);

			\draw[a         ] (d5) -- (b2);
			\draw[very thick] (d7) -- (c3);
			\draw[very thick] (b1) -- (c3);
			\draw[a         ] (d3) -- (b1);
			\draw[very thick] (c1) -- (d4);
		\end{tikzpicture}
        \vspace{-3ex}
		\caption{Auxiliary trees}
		\label{fig:auxtrees}
	\end{subcaptionblock}
	\begin{subcaptionblock}{.2\textwidth}
		\begin{tikzpicture}
			\node[o] (p1)               {};
			\node[o] (p2) [right=of p1] {};
			\node[o] (p3) [right=of p2,xshift=-15pt] {preferred edge};

			\node[o] (n1) [below=of p1] {};
			\node[o] (n2) [right=of n1] {};
			\node[o] (n3) [right=of n2,xshift=-15pt] {normal edge};

			\node[o] (x1) [below=of n1] {};
			\node[o] (x2) [right=of x1] {};
			\node[o] (x3) [right=of x2,xshift=-15pt] {path parent};

			\draw[very thick] (p1) -- (p2);
			\draw             (n1) -- (n2);
			\draw[a]          (x1) -- (x2);
		\end{tikzpicture}
		\vspace{1cm}
	\end{subcaptionblock}
    \vspace{-2ex}
	\caption{Indexed trie: ordered trie where paths are indexed with auxiliary splay trees}
	\label{fig:trie}
\end{figure}

%% file: eval.tex
\section{Evaluation}
\label{sec:eval}

\input{fig/bench}

All the concepts introduced in this paper have been implemented in C++ as part of the \mimir framework.

\paragraph{Benchmarks}

To evaluate the scalability of our algorithms and data structures, we consider three synthetic benchmarks specifically constructed to stress-test their asymptotic behavior:
\begin{enumerate*}
	\item \emph{Loop Cascade} consists of $n$ consecutive counting loops, where each loop’s bound depends on the final counter value of the previous loop.
	      The loops are expressed in \ac{CPS} (similar to \cref{fig:loops}) and embedded in a \ac{CPS} function that provides the initial bound and a return continuation, which is invoked at the end with the last loop’s counter.
	      While not semantically meaningful, this benchmark captures the data-dependency patterns typical of sequential computations where each step depends on the result of the previous one.
	\item \emph{Accumulating Loop Cascade} extends this setup by summing all loop counters and passing the accumulated result to the return continuation.
          This modification keeps all loop counters live until the end, requiring their propagation as free variables from the final continuation back to their loop headers.
          This is the worst-case scenario for programs without loop nests:
          we introduce a number of variables proportional to the program size and keep them alive until the end.
	\item \emph{Loop Nest} is similar to \emph{Loop Cascade}, but all loops are nested rather than consecutive.
	      As a result, all loop counters must be propagated as free variables through every enclosing loop up to their defining headers.
          Thus, all variables must be propagated as in the \emph{Acc.~Loop Cascade}.
          This benchmark represents the worst case: it introduces a number of nested loops and variables proportional to the program size.
\end{enumerate*}

\paragraph{Algorithms}

For each program, we measured the execution time of the following operations:
\begin{enumerate*}
	\item Computation of \emph{free variables} for all functions in the generated program.
	\item A \emph{$\beta$-reduction} on the container function using a dummy value, producing a new function in which all occurrences of the original variable are replaced by the dummy.
        The free variables have already been computed in the previous step and are already cached.
        This operation particularly stresses the set intersection routines.
	\item Construction of the \emph{nesting tree}, including sibling dependencies, and \acp{SCC}.
\end{enumerate*}

\paragraph{Data Structures}

To compare our custom indexed trie\footnote{%
    Including the optimizations discussed in \cref{sec:impl}; in particular, we use ordered, hash-consed arrays for small sets.%
}
used for tracking sets of variables and functions, we also implemented a reference implementation that uses \lst|std::set| instead.
In addition, we ported \mimir to use the C++ library \immer~\cite{DBLP:journals/pacmpl/Puente17} (version 0.8.1).
\immer provides persistent, immutable data structures, including a persistent set type that we employed as an alternative to the indexed trie.
It also supports transient (mutable) operations that allow efficient batch updates before reverting to a persistent form---a feature we use extensively, for example, during free-variable computations.

\paragraph{Comparison with \llvm}

We compiled all \mimir benchmark programs to \llvm and compared \llvm's dominance with \mimir's free-variable analysis, as well as \llvm's inliner with \mimir's $\beta$-reduction.
We emphasize that these comparisons are not apples-to-apples: \mimir's IR is strictly more expressive than \llvm's, supporting first-class functions, polymorphism, dependent types, and plugin-defined abstract data types---which necessarily increases both \ac{IR} size and traversal cost.

\paragraph{Setup}

All experiments were conducted on a CachyOS system equipped with an AMD Ryzen AI 9 HX PRO 370 (12 physical cores, 24 logical processors, up to \SI{5.16}{GHz} boost) and \SI{32}{GiB} LPDDR5-7500\,MT/s memory (effective clock ${\sim}\SI{3750}{MHz}$).
CPU frequency scaling was disabled (\verb|performance| governor).
All measurements were single-threaded.
Each benchmark ran on an otherwise idle system, with processes pinned to a single core via \verb|taskset|.
All C++ code was compiled with GCC 15.2.1 using \verb|-O3 -march=native -DNDEBUG|.
Timings were measured via \lst|std::chrono::steady_clock|.
We measured \llvm using \lst|opt| (version 22.1.1) by recording wall-clock time.
These measurements are less precise than those obtained for \mimir.
Consequently, the lower end of the \llvm plots is either missing or appears noisy.

\subsection{Discussion}

\cref{fig:bench} summarizes the results in double logarithmic scale.
Each row corresponds to one benchmark strategy.
The x-axis shows the number of generated loops, doubling at each tick.
Each column corresponds to one of the measured operations.
The y-axis reports the median execution time in microseconds (\si{\micro\second}) across nine runs, following a warm-up of three runs.
The shaded regions indicate the observed range of measurements.

The indexed trie outperformed \immer in most cases, except for nesting tree construction on \emph{Loop Nest}, where \immer performed better at larger input sizes.
Often, the indexed trie was about twice as fast and, in some cases, around four times faster.
Moreover, \immer exhibited significantly higher memory consumption, causing out-of-memory termination of the \emph{Acc.~Loop Cascade} benchmark for programs exceeding $2^{13}$ loops.
The \lst|std::set| implementation roughly tracked \immer's performance in about half of the cases, but was significantly slower in the other half.
Because \lst|std::set| does not share any data, not only \emph{Acc.~Loop Cascade} but also \emph{Loop Cascade} ran out of memory for input sizes exceeding $2^{13}$ loops.
The indexed trie delivers the best overall performance; therefore, we focus on this data structure for the remainder of this discussion.

The measured \mimir operations exhibit $n \log n$ scaling for both \emph{Loop Cascade} and \emph{Acc.~Loop Cascade}.
\emph{Loop Nest} produces programs whose loop connectedness grows proportionally with program size:
Each iteration of the fixed-point computation for free variables propagates through only a single nesting level, the indexed trie increases in depth, and nesting tree construction must traverse increasingly deep hierarchies to locate parent nodes.
For smaller inputs, performance is dominated by hash-consed arrays; the asymptotic behavior becomes apparent only for larger instances.
Empirically, the observed behavior at the high end is consistent with $\mathcal{O}(n^3 \log^3 n)$ scaling for free-variable computation and $\mathcal{O}(n^2 \log^2 n)$ scaling for $\beta$-reduction and nesting tree construction.
To aid visual comparison, we include faint reference curves; these are shifted along the x-axis in the last row to align with the measurements.
Due to this growth, we limited the \emph{Loop Nest} benchmarks to $2^{11}$ loops.
In this case, computing the free variables took~${\sim}\SI{3}{\minute}$.

Deeply nested loops exceeding hundreds or thousands of levels are essentially absent in real-world software.
For instance, within the entire SPEC2006 benchmark suite, the maximum observed loop-nesting depth is only~22 \cite[see][Table~1]{Nie2021LoopSelection}.
We therefore conclude that the algorithms and data structures presented in this paper scale efficiently to programs with millions of data dependencies, except in contrived cases with unrealistically deep nesting.

\llvm's dominance analysis is several times faster than \mimir's free-variable computation, as expected:
\llvm uses a near linear algorithm that does not depend on the loop connectedness;
for this reason, \llvm's analysis remains efficient even for large \emph{Loop Nests}.
\mimir's free-variable computation, on the other hand, naturally extends to higher-order programs and is highly local and incremental, whereas \llvm must recompute the dominator tree on a per-function basis whenever the \ac{CFG} changes.
However, it is difficult to measure how well these advantages outweigh a much faster dominance analysis.
\mimir's free‑variable analysis can still be considered fast in absolute terms.
For reference, computing free variables for more than three million functions comprising over one million ($2^{20}$) loops takes less than~\SI{4}{\second}.

Despite operating on a richer \ac{IR}, $\beta$-reduction in \mimir is consistently faster than inlining in \llvm with the exception of large \emph{Loop Nests}.
Furthermore, \mimir's \ac{SoN} representation gives rise to several optimizations:
constant folding, arithmetic simplification, semi-global value numbering via hash-consing\footnote{Semi-global in the sense that expressions may float across function/basic-block boundaries while variables remain distinct.}, as well as dead-code and unreachable-block elimination through graph traversal.
To approximate these effects in \llvm, we additionally ran \texttt{instcombine}, \texttt{early-cse}, \texttt{dce}, and \texttt{unreachableblockelim} after inlining.
Applying these optimizations further widens the gap.

\pgfplotstableread[header=true]{data/regex.dom.merged}\datadom
\pgfplotstableread[header=true]{data/regex.inl.merged}\datainl
\pgfplotstableread[header=true]{data/regex.opt.merged}\dataopt
\pgfplotstableread[header=true]{data/trie.fvs.regex.merged}\datafvs
\pgfplotstableread[header=true]{data/trie.beta.regex.merged}\databeta

\pgfplotstablegetelem{0}{median}\of{\datadom}  \edef\vardom{\pgfplotsretval}
\pgfplotstablegetelem{0}{median}\of{\datainl}  \edef\varinl{\pgfplotsretval}
\pgfplotstablegetelem{0}{median}\of{\dataopt}  \edef\varopt{\pgfplotsretval}
\pgfplotstablegetelem{0}{median}\of{\datafvs}  \edef\varfvs{\pgfplotsretval}
\pgfplotstablegetelem{0}{median}\of{\databeta} \edef\varbeta{\pgfplotsretval}

To evaluate how \mimir and \llvm compare on real-world programs with large control flow graphs, we used \mimir's \texttt{regex} plugin to compile a large regular expression matching a single HTTP access-log line, incorporating extensive enumerations of HTTP methods, status codes, timestamps, and path characters.
The resulting matcher contains 1058 continuations in \mimir and the same number of basic blocks when compiled to \llvm.
Free-variable analysis took \SI{\varfvs}{\micro\second}, while \llvm's dominance analysis completed in ${\sim}\SI{\vardom}{\micro\second}$;
$\beta$-reduction required \SI{\varbeta}{\micro\second}, whereas inlining took ${\sim}\SI{\varinl}{\micro\second}$ (${\sim}\SI{\varopt}{\micro\second}$ with optimizations).
These numbers confirm that \llvm's dominance analysis is several times faster than \mimir's free-variable analysis while $\beta$-reduction is competitive.

%% file: fig/bench.tex
\newcommand{\myplot}[3]{%
    \IfFileExists{#3}{%
        \addplot[name path=upper, draw=none, forget plot] table [x index=0, y index=2, col sep=space] {#3}; 
        \addplot[name path=lower, draw=none, forget plot] table [x index=0, y index=3, col sep=space] {#3}; 
        \addplot[fill=#1!20, forget plot] fill between[of=upper and lower];
        \addplot[#1, #2, mark options={scale=0.5}] table [x index=0, y index=1, col sep=space] {#3};
    }{}%
}
\colorlet{ACMPurple2}{ACMPurple!50}
\begin{figure}[p]
	\begin{tikzpicture}
		\begin{groupplot}[
				group style={
						group size=3 by 3,
						horizontal sep=0cm,
						vertical sep=0cm,
						xlabels at=edge bottom,
						ylabels at=edge left,
						xticklabels at=edge bottom,
						yticklabels at=edge left,
					},
                title style={font=\footnotesize},
                tick align=outside,
                tick label style={font=\footnotesize, inner sep=0pt},
				width=.42\textwidth,   
				height=0.51\textwidth,
				log basis x={2},
				log basis y={10},
				grid=both,
				title style={yshift=-1.5ex},
                samples=100,
                domain=1:1048576,
                xmode=log,
				xmin=1, xmax=1048576, 
				xtickten={4,8,...,20}, 
				xticklabel={\(2^{\pgfmathprintnumber{\tick}}\)},
				extra x ticks={1,2,4,8,16,32,64,128,256,512,1024,2048,4096,8192,16384,32768,65536,131072,262144,524288,1048576}, 
				extra x tick style={grid style={gray!15}, 
						tick label style={opacity=0}, 
					},
                ymode=log,
                ylabel style={font=\footnotesize,yshift=-2.5ex},
				ymin=1, ymax=1000000000, 
				ytickten={3,6,...,12}, 
				yticklabel={\(10^{\pgfmathprintnumber{\tick}}\)},
				extra y ticks/.expanded={
                        1, 10, 100, 1000, 10000, 100000, 1000000, 10000000, 100000000, 1000000000
					},
				extra y tick style={
						grid style={gray!25},
						tick label style={opacity=0},
					},
                legend pos=north west,
                legend style={font=\footnotesize},
			]

			\nextgroupplot[title={Free Variables}, ylabel={Loop Cascade}]

            \addplot[ACMGreen!60, smooth, forget plot] {x * ln(x) / ln(2)};
            \label{nlogn}

			\myplot{ACMBlue}{mark=*,ultra thick}{data/trie.fvs.0.merged};
            \addlegendentry{Indexed Trie}

			\myplot{ACMOrange}{mark=square*}{data/immer.fvs.0.merged};
            \addlegendentry{\immer}

			\myplot{ACMYellow}{mark=triangle*,thin}{data/set.fvs.0.merged};
            \addlegendentry{\texttt{std::set}}

			\myplot{ACMPurple}{densely dashed,mark=diamond*}{data/0.dom.merged};
            \addlegendentry{\llvm dominance}

			\nextgroupplot[title={$\beta$-Reduction}]

            \addplot [ACMGreen!60, smooth] {x * ln(x) / ln(2)};

			\myplot{ACMBlue}{mark=*,ultra thick}{data/trie.beta.0.merged};
			\myplot{ACMOrange}{mark=square*}{data/immer.beta.0.merged};
			\myplot{ACMYellow}{mark=triangle*,thin}{data/set.beta.0.merged};
			\myplot{ACMPurple2}{densely dashed,mark=diamond*}{data/0.inl.merged};
            \label{llvm-inl}
			\myplot{ACMPurple}{densely dashed,mark=diamond*}{data/0.opt.merged};
            \label{llvm-opt}

            \node [draw,fill=white] at (rel axis cs: 0.39,0.9) {\shortstack[l]{
            \footnotesize\ref{llvm-opt} \llvm inline + opt \\
            \footnotesize\ref{llvm-inl} \llvm inline}};

			\nextgroupplot[title={Nesting Tree}]

            \addplot [ACMGreen!60, smooth] {x * ln(x) / ln(2)};

			\myplot{ACMBlue}{mark=*,ultra thick}{data/trie.nest.0.merged};
			\myplot{ACMOrange}{mark=square*}{data/immer.nest.0.merged};
			\myplot{ACMYellow}{mark=triangle*,thin}{data/set.nest.0.merged};

			\nextgroupplot[ylabel={Acc. Loop Cascade}]

            \addplot [ACMGreen!60, smooth] {x * ln(x) / ln(2)};

			\myplot{ACMBlue}{mark=*,ultra thick}{data/trie.fvs.1.merged};
			\myplot{ACMOrange}{mark=square*}{data/immer.fvs.1.merged};
			\myplot{ACMYellow}{mark=triangle*,thin}{data/set.fvs.1.merged};
			\myplot{ACMPurple}{densely dashed,mark=diamond*}{data/1.dom.merged};

			\nextgroupplot

            \addplot [ACMGreen!60, smooth] {x * ln(x) / ln(2)};

			\myplot{ACMBlue}{mark=*,ultra thick}{data/trie.beta.1.merged};
			\myplot{ACMOrange}{mark=square*}{data/immer.beta.1.merged};
			\myplot{ACMYellow}{mark=triangle*,thin}{data/set.beta.1.merged};
			\myplot{ACMPurple2}{densely dashed,mark=diamond*}{data/1.inl.merged};
			\myplot{ACMPurple}{densely dashed,mark=pentagon*}{data/1.opt.merged};

			\nextgroupplot

            \addplot[ACMGreen!60, smooth] {x * ln(x) / ln(2)};

			\myplot{ACMBlue}{mark=*,ultra thick}{data/trie.nest.1.merged};
			\myplot{ACMOrange}{mark=square*}{data/immer.nest.1.merged};
			\myplot{ACMYellow}{mark=triangle*,thin}{data/set.nest.1.merged};

			\nextgroupplot[ylabel={Loop Nest}]

            \addplot[ACMRed!40, smooth,domain=1:1024,xshift=29] {(x * ln(x) / ln(2))^3};
            \label{nlogn3}

			\myplot{ACMBlue}{mark=*,ultra thick}{data/trie.fvs.2.merged};
			\myplot{ACMOrange}{mark=square*}{data/immer.fvs.2.merged};
			\myplot{ACMYellow}{mark=triangle*,thin}{data/set.fvs.2.merged};
			\myplot{ACMPurple}{densely dashed,mark=diamond*}{data/2.dom.merged};

			\nextgroupplot

			\myplot{ACMBlue}{mark=*,ultra thick}{data/trie.beta.2.merged};
			\myplot{ACMOrange}{mark=square*}{data/immer.beta.2.merged};
			\myplot{ACMYellow}{mark=triangle*,thin}{data/set.beta.2.merged};
			\myplot{ACMPurple2}{densely dashed,mark=diamond*}{data/2.inl.merged};
			\myplot{ACMPurple}{densely dashed,mark=pentagon*}{data/2.opt.merged};

            \plot [ACMOrange!40, smooth,domain=1:2048,xshift=30] {(x * ln(x) / ln(2))^2};
            \label{nlogn2}

			\nextgroupplot

            \addplot[ACMOrange!40, smooth,domain=1:4096,xshift=33] {(x * ln(x) / ln(2))^2};

			\myplot{ACMBlue}{mark=*,ultra thick}{data/trie.nest.2.merged};
			\myplot{ACMOrange}{mark=square*}{data/immer.nest.2.merged};
			\myplot{ACMYellow}{mark=triangle*,thin}{data/set.nest.2.merged};

            \node [draw,fill=white] at (rel axis cs: 0.75,0.175) {\shortstack[l]{
            \footnotesize\ref{nlogn} $n \log_2 n$ \\
            \footnotesize\ref{nlogn2} $n^2 \log_2^2 n$ \\
            \footnotesize\ref{nlogn3} $n^3 \log_2^3 n$}};

		\end{groupplot}
        \node[anchor=north east, rotate=90, align=center]
        at ($(group c1r3.south west) + (-2.5em,4em)$) {\footnotesize \si{\micro\second} (microsec.) $\to$};
        \node[anchor=north, align=center]
        at ($(group c1r3.south west) + (1em,-1em)$) {\footnotesize \# loops $\to$};
	\end{tikzpicture}
    \vspace{-6ex}
    \caption{%
        Performance (in \si{\micro\second}, lower is better) of three operations---free variable computation, $\beta$-reduction, and nesting tree computation---on three synthetic programs with increasing numbers of loops.
        The plots compare three different set implementations: the indexed trie, \immer, and \texttt{std::set}.
        Additionally, free variable computation is compared against LLVM's dominance analysis;
        $\beta$-reduction is compared against LLVM's inliner (both with and without a sequence of \texttt{instcombine}, \texttt{early-cse}, \texttt{dce}, \texttt{and unreachableblockelim} to approximate \mimir's on-the-fly \acl{SoN} effects).
        For reference, \SI{e3}{\micro\second} = \SI{1}{\milli\second}, \SI{e6}{\micro\second} = \SI{1}{\second}.%
    }
    \label{fig:bench}
\end{figure}

%% file: relwork.tex
\section{Related Work}
\label{sec:relwork}

\paragraph{Binders}

\lamG identifies a function with the variable it introduces, so both share the same identifier.
This removes the need to synchronize separate representations and avoids an entire class of bookkeeping issues.
Still, as \citet{cockx2021syntax} observes,
\enquote{there are countless different techniques, frameworks, and meta-languages for dealing with name binding, none of which come close to be universally accepted or clearly superior to the others.}
Because \lamG is a graph-based representation that supports mutation and (mutual) recursion via its cyclic structure rather than through an explicit fixed-point operator, many of the alternatives discussed by \citeauthor{cockx2021syntax}, such as (Co-)De Bruijn indices, are not suitable.

\paragraph{SSA Form}

\ac{SSA} form was introduced by \citet{10.1145/73560.73562} and achieved widespread adoption in compiler design following the publication of an efficient construction algorithm by \citet{DBLP:journals/toplas/CytronFRWZ91}.
\Cref{prop:dom} is not stated as an explicit constraint; rather, it is guaranteed by that construction.
Later work~\cite[e.g.][]{DBLP:journals/spe/BriggsCHS98,DBLP:conf/cc/HackGG06} formalized this guarantee as the \emph{strict \ac{SSA} property}.
Without it, a use of a variable may appear in a block that is \emph{not} dominated by its definition, causing standard analyses and transformations to break down.
Consequently, modern \acp{IR} such as \llvm~\cite{DBLP:conf/cgo/LattnerA04} and \mlir~\cite{DBLP:conf/cgo/LattnerAB0DPRSV21} treat dominance as integral to \ac{SSA} well-formedness.

\mlir introduces a degree of higher-order behavior through regions that can be passed syntactically to operations specifically designed to consume regions.
This complicates dominator-tree construction~\cite{BhatNiu2023ExtendingDominanceToMLIRRegions}.
\mlir maintains a purely structural representation: regions resemble functions, may contain blocks, and can refer to free variables from enclosing scopes.
Accordingly, \mlir embeds each region's dominator tree into the parent \ac{CFG} hierarchy.
This may overapproximate data dependencies and place the region's dominator tree unnecessarily deep.
In contrast, our nesting tree computes the \emph{minimal} required nesting purely from free variables, even in higher-order settings.
Furthermore, \mlir is not a faithful higher-order representation, as it cannot directly represent the program in \cref{fig:iter} without at least partially lowering the program.

\paragraph{Dominance in SSA}

While efficient algorithms exist to compute the dominator tree~\cite{DBLP:journals/toplas/LengauerT79,cooper06}, the tree is highly brittle: any change to the \ac{CFG} requires recomputation.
Since the dominator tree is ubiquitous in \ac{SSA}-based compilers, it must be frequently rebuilt whenever the \ac{CFG} changes.
This led compiler engineers to implement incremental updates to the dominator tree, avoiding full per-function recomputation~\cite{Kuderski2017DominatorTrees}.
Engineers must still manually track \ac{CFG} mutations and keep the framework in sync.
Our work sidesteps this entirely by asking how functions are \emph{nested} and automatically maintaining free-variable sets in a transparent way, without manual intervention.

While this paper advocates free variables rather than dominance as the primary structural notion, dominance remains valuable in some contexts (e.g., classical \ac{SSA} construction before \ac{SSA} variables exist).
There are also algorithms that construct \ac{SSA} without dominance information~\cite{DBLP:conf/cc/BraunBHLMZ13,DBLP:journals/pacmpl/Lemerre23}.

\paragraph{Explicit Scoping}

Schneider's verified compiler LVC~\cite{DBLP:phd/dnb/Schneider18d} recovers SSA-style benefits in the lexically-scoped term language IL by replacing dominance with a syntactic \emph{coherence} predicate that guarantees agreement between its imperative (IL/I) and functional (IL/F) semantics.
\citet{DBLP:journals/pacmpl/Lemerre23} \enquote{showed that we can also have both interpretations on a graph-based language without an explicit syntactic scope}.
We are currently working on an \ac{SSA} construction in \mimir that extends \citeauthor{DBLP:journals/pacmpl/Lemerre23}'s work for higher-order programs.

\citet{DBLP:conf/icfp/Kennedy07} presents a \ac{CPS}-based \ac{IR} that represents the term structure as a doubly-linked tree, in which each subterm maintains an explicit up-link to its immediately enclosing parent.
This design enables constant-time replacements, deletions, and insertions.
Variable uses are maintained in a doubly-linked circular list, with same-binder equivalence classes computed via a union-find structure, enabling near-constant-time substitution when replacing one variable with another.
However, the representation remains explicitly scoped and may still require block floating (see \cref{sec:intro}).
Moreover, free-variable computation requires global re-analysis after every transformation.
Given the destructive nature of the updates, there is no opportunity to precompute and incrementally maintain partial results as our framework does.

\paragraph{Sea of Nodes}

A \ac{SoN}-style \ac{IR}~\cite{DBLP:journals/toplas/ClickC95} unifies data-flow and control-flow dependencies into a single graph-based representation, eliminating the rigid instruction ordering found in traditional \ac{CFG}-oriented \acp{IR}.
This approach has roots in the global value graph~\cite{DBLP:journals/jcss/ReifL86} and cyclic term graph~\cite{DBLP:journals/fuin/AriolaK96}.
It enables aggressive global optimizations by representing operations as nodes constrained only by their dependencies.
\citet{DBLP:conf/cc/DemangeRP18} describe a formal semantics for \ac{SoN}.

\paragraph{Soup of Functions}

\thorin~\cite{DBLP:conf/cgo/LeissaKH15} pioneered a scopeless, higher-order \ac{CPS} \ac{IR} with primitive operations in direct style, while all expressions are managed in a \ac{SoN}.
\mimir evolved from \thorin, retaining its higher-order \ac{SoN} representation, and extends it to support direct-style functions, polymorphism, dependent types, and a plugin system~\cite[e.g.][]{DBLP:conf/cc/UllrichHL25}.
It also served as the implementation platform for the concepts presented in this paper.
\enquote{CPS soup}~\cite{wingolog2023approaching} in \guile is a \ac{CPS}-based \ac{IR} that uses integers as labels for all continuations and maintains a flat, persistent map from labels to continuations, similar to a \lamG program.
\thorin and \mimir instead use the memory addresses of functions as labels.
In contrast to \guile and \thorin, \lamG and \mimir also support direct style.
Furthermore, unlike \mimir, \guile represents non-continuation expressions using traditional expression trees.
\thorin, \guile, and \mimir (prior to our work) overlay the soup of functions with a derived \ac{CFG} (similar to \irefrule{Succ}) to compute properties such as dominance.
This derived \ac{CFG}, however, must be restricted to the \emph{reachable and live} subset of the program, since neither unused continuations nor external callees (e.g., \lst|printf|) should be treated as \ac{CFG} nodes.
For this reason, these frameworks typically perform dedicated reachability and liveness analyses before applying $\beta$-reduction or similar transformations.
Moreover, dominance still overapproximates the actual data dependencies, as discussed in \cref{sec:overview}.
This paper advocates a paradigm shift by instead relying on free-variable information and avoiding a \ac{CFG} overlay.
In our setting, classical liveness corresponds exactly to free-variable information \cite[p.~22]{DBLP:books/cu/Appel1992}.
Put another way, our framework tracks the liveness of variables, where each variable corresponds to the set of $\phi$-functions introduced in the same basic block.
Our new $\beta$-reduction algorithm, for example, visits only reachable functions without a prior global analysis, and uses free-variable information to decide locally which expressions must be specialized.
We introduce a free-variable framework for \enquote{soup}-style \acp{IR}, develop a metatheory for this setting, and introduce nesting as a relaxed alternative to dominance.

\paragraph{Data Structures}

We need an immutable data structure to keep track of sets of variables and functions.
These data structures have a long history in functional programming and compiler implementation.
By preserving previous versions under update, they enable persistent program representations that support efficient incremental queries and safe sharing across transformations~\cite{DBLP:books/daglib/0097014}.
Techniques such as path copying, hash-consing, and trie-based structures provide logarithmic or even amortized constant-time updates while avoiding destructive mutation.
These properties make immutable data structures attractive for compilers and program analyses that require scalable maintenance of auxiliary information.

A link-cut tree~\cite{DBLP:journals/jcss/SleatorT83} is a lesser-known data structure that proved highly useful in this work.
Given the frequent occurrence of tree structures in compilers, we suspect that link-cut trees could be applied more broadly to accelerate various algorithms.
For instance, they have already been used to speed up navigation of relational algebra trees in databases~\cite{DBLP:journals/pvldb/FentM023}.

The sibling dependencies bear a strong resemblance to the join-edges in a DJ-graph~\cite{DBLP:conf/popl/SreedharG95}.
In both cases, additional edges enrich a tree (the dominator tree or, in our case, the nesting tree) to capture non-hierarchical relationships.
However, while the DJ-graph records control-flow joins, sibling dependencies capture data-flow relationships among siblings, making them applicable beyond first-order CFGs and suitable for higher-order or non-control-based \acp{IR}.

%% file: concl.tex
\section{Conclusion}
\label{sec:concl}

This paper revisits one of the central assumptions underlying \ac{SSA} form:
the reliance on dominance as the organizing principle for relating definitions and uses.
While dominance has proven effective for first-order \acp{CFG}, it becomes increasingly imprecise and, in some cases, ill-defined in the presence of higher-order abstractions.

To address this, we introduced \lamG, a graph-based intermediate representation that replaces \ac{CFG}-based dominance with nesting, a structural property derived from free-variable information.
By treating free variables as the primary source of structural information, we extend classical \ac{SSA}-style transformations, such as substitution and critical-edge elimination, to a higher-order, scopeless setting.
Furthermore, we showed that this paradigm shift is practical: using an incremental caching framework and an indexed trie data structure, free-variable queries, dependency-aware $\beta$-reduction, and nesting-tree construction scale log-linearly for the benchmarks representative of real-world workloads.

More broadly, this work suggests that control-flow structure need not be the primary lens through which compiler analyses are expressed.
By grounding program structure in data dependencies instead, we obtain a representation that more closely reflects the semantics of modern programming languages and enables optimizations and analyses that are difficult to express in traditional \ac{SSA} frameworks.